\RequirePackage{fix-cm}
\documentclass[smallcondensed]{svjour3} 
\smartqed
\usepackage{amsmath}
\usepackage{mathrsfs,amssymb}
\usepackage{lineno,hyperref}
\newcommand{\hs}{\hspace}
\newcommand{\hh}{\hspace{0.5mm}}
\newcommand{\ov}{\overline}
\newcommand{\ti}{\tilde}
\newcommand{\mR}{\mathbb{R}}

\newcommand{\mP}{\mathbb{P}}
\newcommand{\mE}{\mathbb{E}}
\newcommand{\mL}{\mathbb{L}}
\newcommand{\mD}{\mathbb{D}}
\newcommand{\mF}{\mathscr{F}}
\newcommand{\mA}{\mathcal{A}}
\newcommand{\nn}{\nonumber}
\newcommand{\Si}{\Sigma}
\newcommand{\si}{\sigma}
\newcommand{\ep}{\varepsilon}
\newcommand{\de}{\delta}
\newcommand{\lam}{\lambda}
\newcommand{\ph}{\varphi}
\newcommand{\pa}{\partial}
\newcommand{\lab}{\left|}
\newcommand{\rab}{\right|}
\newcommand{\lno}{\left| \hspace{-0.3mm} \left|}
\newcommand{\rno}{\right| \hspace{-0.3mm} \right|}
\newcommand{\la}{\left(}
\newcommand{\ra}{\right)}

  \makeatletter
    
    \@addtoreset{equation}{section}
  \makeatother
\def\Xint#1{\mathchoice
 {\XXint\displaystyle\textstyle{#1}}%
 {\XXint\textstyle\scriptstyle{#1}}%
 {\XXint\scriptstyle\scriptscriptstyle{#1}}%
 {\XXint\scriptscriptstyle\scriptscriptstyle{#1}}%
 \!\int}
 \def\XXint#1#2#3{{\setbox0=\hbox{$#1{#2#3}{\int}$}
 \vcenter{\hbox{$#2#3$}}\kern-.5\wd0}}
 
 \def\dashint{\Xint-}
\modulolinenumbers[5]
\allowdisplaybreaks[2]

\begin{document}

\title{A stochastic conservation law with nonhomogeneous Dirichlet boundary conditions}


\author{Kazuo Kobayasi \and Dai Noboriguchi}


\institute{K. Kobayasi 
\at Department of Mathematics, Education and Integrated Arts and Sciences, Waseda University, 1-6-1 Nishi-Waseda, Shinjuku-ku, Tokyo, 169-8050, Japan  \\
\email{kzokoba@waseda.jp}
\and
D. Noboriguchi 
\at Graduate School of Education, Waseda University, 1-6-1 Nishi-Waseda, Shinjuku-ku, Tokyo, 169-8050, Japan \\
\email{588243-dai@fuji.waseda.jp}}

\date{ \mbox{} \\[-4mm] \noindent Accepted for publication in Acta Mathematica Vietnamica \\ Received: date / Accepted: date}

\maketitle

\begin{abstract}
This paper discusses the initial-boundary value problem (with a nonhomogeneous boundary condition) for a multi-dimensional scalar first-order conservation law with a multiplicative noise. One introduces a notion of kinetic formulations in which the kinetic defect measures on the boundary of a domain are truncated. In such a kinetic formulation one obtains a result of uniqueness and existence. The unique solution is the limit of the solution of the stochastic parabolic approximation.
\end{abstract}

\section{Introduction}
In this paper we study the first order stochastic conservation law of the following type
\begin{eqnarray}
du + {\rm div} (A(u)) dt = \Phi (u) dW(t) \hs{5mm} {\rm in } \hs{2mm} \Omega \times Q, \label{SSCL1}
\end{eqnarray}with the initial condition
\begin{eqnarray}
u(0 , \cdot ) = u_0(\cdot) \hs{5mm} {\rm in} \hs{2mm} \Omega \times  D, \label{SSCL2}
\end{eqnarray}and the formal boundary condition
\begin{eqnarray}
 ``u = u_b\mbox{''} \hs{5mm} {\rm on } \hs{2mm}  \Omega \times \Si. \label{SSCL3}
\end{eqnarray}Here $D \subset \mR^d$ is a bounded domain with a Lipschitz boundary $\pa D$, $T > 0$, $Q = (0,T) \times D$, $\Si = (0,T) \times \pa D$ and $W$ is a cylindrical Wiener process defined on a stochastic basis $(\Omega , \mF , (\mF_t) , P)$. More precisely, $(\mF_t)$ is a complete right-continuous filtration and $ W(t) = \sum_{k=1}^{\infty} \beta_k (t) e_k$ with $(\beta_k)_{k \geq 1}$ being mutually independent real-valued standard Wiener processes relative to $(\mF_t)$ and $(e_k)_{k \geq 1}$ a complete orthonormal system in a separable Hilbert space $H$ (cf. \cite{c} for example). Our purpose of this paper is to present a definition of kinetic solution to the initial-boundary value problem (\ref{SSCL1})-(\ref{SSCL3}) and to prove a result of uniqueness and existence of such a solution. \\
\indent In the deterministic case of $\Phi = 0$, the problem has been extensively studied. It is well known that a smooth solution is constant along characteristic lines, which can intersect each other and shocks can occur. Moreover, when the characteristic intersects both $\{ 0 \} \times D$ and $\Si$, the problem (\ref{SSCL1})-(\ref{SSCL3}) would be overdetermind if (\ref{SSCL3}) were assumed in the usual sense. Thus, an appropriate frameworks of entropy solutions, together with entropy-boundary conditions, has been considered to obtain uniqueness and existence results. In the BV setting Bardos, Le Roux and N${\rm \acute{e}}$d${\rm \acute{e}}$lec \cite{f} first gave an interpretation of the boundary condition (\ref{SSCL3}) as an "entropy" inequality on $\Si$, which is the so-called BLN condition, and proved the well-posedness of the initial-boundary value problem. Otto \cite{i} extended it to the $L^\infty$ setting by introducing the notion of boundary entropy-flux pairs. Imbert and Vovelle \cite{b} gave a kinetic formulation of weak entropy solutions of the initial-boundary value problem and proved the uniqueness of such a kinetic solution. Concerning deterministic degenerate parabolic equations, see \cite{o} and \cite{p}. \\
\indent To add a stochastic forcing $\Phi (u) dW(t)$ is natural for applications, which appears in wide variety of fields as physics, engineering and others. The Cauchy problem for the stochastic conservation law (\ref{SSCL1}) with additive noise has been studied in \cite{l}, with multiplicative noise in \cite{d}, where the uniqueness of the strong entropy solution is proved in any dimension, the existence in one dimension. Also see \cite{n} for the existence of strong entropy solutions in any dimension.\\
\indent By using a kinetic formulation the well-posedness for kinetic solution to scalar conservation laws with a general multiplicative noise in a $d$-dimensional torus was obtained by Debussche and Vovelle \cite{a}. The main advantage from using kinetic formulations developed by Lions, Perthame and Tadmor for deterministic case \cite{e} is that the formulation keeps track of the dissipation of noise by solutions and works in the $L^1$ setting. Those results have been extended to the case of degenerated parabolic stochastic equations. in \cite{t} \\
\indent There are a few paper concerning the Dirichlet boundary value problem for stochastic conservation laws. Vallet and Wittbold \cite{s} extended the result of Kim \cite{l} to the multi-dimensional Dirichlet problem with additive noise. In the recent paper \cite{k}, Bauzet, Vallet and Wittbold studied the Dirichlet problem in the case of multiplicative noise under the restricted assumption that the flux function $A$ is global Lipschitz. In \cite{s} and \cite{k} the boundary condition is formulated in the sense of Carrillo, which consists in formulating the boundary condition by inequalities involving the semi-Kru${\rm \check{z}}$kov entropies.\\
\indent Our main results are counterparts of the results in \cite{a} in the case of initial-boundary value problems. The flux function $A$ is supposed to have the  bounded second derivatives (see Theorem \ref{existence1} below). Thus, an important example of inviscid Burgers' equation can be included. Moreover, in the homogeneous boundary case, i.e., in the case of Dirichlet's (zero) boundary condition, one can assume only that $A$ is of class $C^2$ and its derivatives have at most polynomial growth (see Theorem \ref{existence2} below). \\
\indent Although the basic idea of the proof is analogous to that of \cite{a} and \cite{b}, the stochastic case is significantly different from the deterministic case. A ``stochastic'' kinetic solution $u$ might blow up at the boundary $\pa D$ even if the data $u_0$, $u_b$ in (\ref{SSCL2}), (\ref{SSCL3}) are bounded. Let us make some more comments on those points. In \cite{b} the defect measures $\bar{m}^{\pm}$ (which are denoted by $m^b_{\pm}$ there) on the boundary $\Si \times \mR_\xi$ play an important role. In particular, it is crucial that $\bar{m}^+$ (resp. $\bar{m}^-$) vanishes for $\xi >>1$ (resp. $\xi <<-1$) in the proof of uniqueness. This property for $\bar{m}^\pm$ comes from the boundedness of the weak entropy solutions. To the contrary, in the stochastic case we have no pathwise $L^\infty$ estimate of kinetic (entropy) solutions $u(t) $ even though both of initial datum $u_0$ and boundary datum $u_b$ belong to $L^\infty$ : It is known only that $\mE \sup_{0\leq t \leq T} \lno u (t) \rno_{L^p (D)}$ is finite for every $p \in [1, \infty )$ and hence we are not able to obtain that the boundary defect measures $\bar{m}^\pm$ vanish as $\xi$ goes to infinity. To overcome this difficulty we introduce a notion of "renormalized" kinetic formulations (Definition \ref{renoma} below), where $\bar{m}^\pm$ are cut off or renormalized on each finite interval $(-N,N)$ of $\mR_\xi$, and we prove the uniqueness of such a renormalized kinetic solution. \\
\indent We now give the precise conditions under which the uniqueness of renormalized kinetic solutions will be proved.
\begin{description}
\item[${\rm (H_1)}$] The flux function $A$: $\mR \to \mR^d$ is of class $C^2$ and its derivatives have at most polynomial growth.
\item[${\rm (H_2)}$] For each $z \in L^2 (D)$, $\Phi (z) : H \to L^2 (D)$ is defined by $\Phi (z) e_k = g_k (\cdot , z (\cdot))$, where $g_k \in C (D \times \mR)$ satisfies the following conditions:
\begin{gather}
 G^2 (x,\xi) = \sum_{k=1}^{\infty} \lab g_k (x,\xi) \rab^2 \leq L (1 + \lab \xi \rab^2), \label{H21} \\
 \sum_{k=1}^{\infty} \lab g_k (x,\xi) - g_k (y , \zeta ) \rab^2 \leq L \la \lab x-y \rab^2 + \lab \xi - \zeta  \rab r(\lab \xi - \zeta \rab) \ra \label{H22}
\end{gather}for every $x,y \in D$, $\xi , \zeta \in \mR$. Here, $L$ is a constant and $r$ is a continuous nondecreasing function on $\mR_+$ with $r(0) = 0$. 
\item[${\rm (H_3)}$] $u_0 \in L^\infty (\Omega \times D )$ and is $\mF_0 \otimes \mathcal{B}(D)$-measurable. $u_b \in L^\infty ( \Omega \times \Si )$ and $\{ u_b (t) \}$ is predictable, in the following sense: For every $p \in [1,\infty)$, the $L^p (\pa D)$-valued process $\{ u_b (t) \}$ is predictable with respect to the filtration $(\mF_t)$.
\end{description}Note that by (\ref{H21}) one has
\begin{gather}
 \Phi : L^2 (D) \to L_2 (H ; L^2 (D)) , 
\end{gather}where $L_2 (H ; L^2 (D))$ denotes the set of Hilbert-Schmidt operators from $H$ to $L^2 (D)$.  \\
\indent The existence of kinetic solutions is proved under more strong conditions than the above ones which will be stated in the beginning of Section 4. \\
\indent This paper is organized as follows. In Section 2, we introduce the notion of kinetic solutions to (\ref{SSCL1})-(\ref{SSCL3}) by using the kinetic formulation, and give some useful lemmas concerning the weak traces on the boundary. In Section 3, we state the $L^1$-contraction (uniqueness) theorem as well as the reduction theorem and prove them. In Section 4, the existence of a kinetic solution is stated and is then proved.

\section{Kinetic solution and generalized kinetic solution}

We give the definition of solution in this section. We mainly follow the notations of \cite{a} and \cite{b}. We choose a finite open cover $\{ U_{\lam_i} \}_{i=0, \ldots , M}$ of $\ov{D}$ and a partition of unity $\{ \lam_i \}_{i=0, \ldots , M}$ on $\ov{D}$ subordinated to $\{ U_{\lam_i} \}$ such that $U_{\lam_0} \cap \pa D = \emptyset$, for $i=1, \ldots , M$,
\begin{gather*}
 D^{\lam_i} : = D \cap U_{\lam_i}  = \{ x \in U_{\lam_i} ; ( \mA_i x )_d > h_{\lam_i} (\ov{\mA_i x}) \} \hs{2mm} {\rm and} \\
 \pa D^{\lam_i} : = \pa D \cap U_{\lam_i}  = \{ x \in U_{\lam_i} ; ( \mA_i x )_d = h_{\lam_i} (\ov{\mA_i x}) \},
\end{gather*}with a Lipschitz function $h_{\lam_i} : \mR^{d-1} \to \mR$, where $\mA_i$ is an orthogonal matrix corresponding to a change of coordinates of $\mR^d$ and $\bar{y}$ stands for $(y_1, \ldots, y_{d-1})$ if $y \in \mR^d$. For the sake of clarity, we will drop the index $i$ of $\lam_i$ and we will suppose that the matrix $\mA_i$ equals to the identity. We also set $Q^\lam = (0,T) \times D^{\lam}$, $\Si^{\lam} = (0,T) \times \pa D^{\lam}$ and $\Pi^{\lam} = \{ \bar{x} ; x \in U^{\lam} \}$.\\
\indent To regularize functions that are defined on $D^{\lam}$ and $\mR$, let us consider a standard mollifier $\psi$ on $\mR$, that is, $\psi$ is a nonnegative and even function in $C^{\infty}_c((-1,1))$ such that $\int_\mR \psi = 1$. We set $\rho^\lam (x) = \Pi_{i=1}^{d-1} \psi(x_i) \psi(x_d - (L_\lam +1))$ for $x = (x_1, \ldots , x_d)$ with the Lipschitz constant $L_{\lam}$ of $h_ \lam$ on $\Pi^{\lam}$. For $\ep , \de > 0$ we set $\rho_\ep^{\lam} (x) = \frac{1}{\ep^d}\rho^\lam (\frac{x}{\ep})$ and $\psi_\de (\xi) = \frac{1}{\de} \psi (\frac{\xi}{\de})$.

\begin{definition}[Kinetic measure] A maps $m$ from $\Omega$ to $\mathcal{M}^+_b ([0,T) \times D \times \mR)$, the set of non-negative finite measures over $[0,T) \times D \times \mR$, is said to be a kinetic measure if
\begin{enumerate}
\renewcommand{\labelenumi}{(\roman{enumi})}
\item $m$ is weak measurable,
\item $m$ vanishes for large $\xi$: if $B_R^c = \{ \xi \in \mR : \lab \xi \rab \geq R \}$ then
 \begin{gather}
  \lim_{R \to \infty} \mE m \la [0,T) \times D \times B_R^c \ra = 0, \label{vanish}
 \end{gather} 
\item for all $\phi \in C_b (D \times \mR)$, the process
 \begin{gather}
  t \mapsto \int_{[0,t] \times D \times \mR} \phi (x, \xi) \hs{0.5mm} d m(s,x,\xi)
 \end{gather}is predictable.
\end{enumerate}
\end{definition}

\begin{definition}[Kinetic solution] \label{renoma} Let $u_0$ and $u_b$ satisfy ${\rm (H_3)}$. A measurable function $u: \Omega \times Q  \to \mR$ is said to be a kinetic solution of (\ref{SSCL1})-(\ref{SSCL3}) if $\{ u(t) \}$ is predictable, for all $p \geq 1$ there exists a constant $C_p \geq 0$ such that for a.e. $t \in [0,T]$,
\begin{gather}
    \lno u(t) \rno_{L^p (\Omega \times D)} \leq C_p, \label{esti23}
\end{gather}there exists a kinetic measure $m$ and if, for any $N>0$, there exist nonnegative $\bar{m}^\pm_N \in L^1 (\Omega \times \Si \times (-N,N) )$ such that $\{ \bar{m}^\pm_N (t) \}$ are predictable, $\bar{m}^+_N (N-0) = \bar{m}^-_N (-N+0) = 0$ for sufficiently large $N$ and $f_+ : = {\bf 1}_{u > \xi}$, $f_- := f_+-1=-{\bf 1}_{u \leq \xi}$ satisfy: for all $\ph \in C^\infty_c ([0,T) \times \ov{D} \times (-N,N) )$,
   \begin{eqnarray}
    & & \hs{-11mm} \int_{Q \times \mR} f_\pm (\pa_t + a(\xi) \cdot \nabla) \ph \hs{0.5mm} d\xi dx dt \nonumber \\
    & &  + \int_{D \times \mR} f^0_\pm \ph(0) \hs{0.5mm} d\xi dx + M_N \int_{\Si \times \mR} f^b_\pm \ph \hs{0.5mm} d\xi d\si (x) dt \nonumber \\
    & & \hs{-11mm} = - \sum_{k = 1}^\infty \int_0^T \int_{D} g_k ( x, u(t,x) ) \hh \ph (x,t,u(t,x)) \hh dx d\beta_k(t) \nonumber \\
    & &  -\frac{1}{2} \int_Q \pa_{\xi} \ph (x,t,u(t,x)) G^2 (x,u(t,x)) \hh dx dt \nonumber \\
    & &  +\int_{[0,T) \times D \times \mR} \pa_{\xi} \ph \hs{0.5mm} dm + \int_{\Si \times \mR} \pa_{\xi} \ph \hh \bar{m}^\pm_N \hh d\xi d\si (x) dt \hs{5mm} {\rm a.s.}, \label{KIFO}
   \end{eqnarray}where $a(\xi) = A'(\xi)$, $M_N = \max_{-N \leq \xi \leq N} \lab a(\xi) \rab$. In (\ref{KIFO}), $f^0_+ = {\bf 1}_{u_0 > \xi}$, $f^b_+ = {\bf 1}_{u_b > \xi}$, $f^0_- = f^0_+ -1$ and $f^b_- = f_+^b -1$.
\end{definition}

For the sake of the proof of existence of kinetic solution, it is useful to introduce the notion of generalized kinetic solution. We start with the definition of kinetic function.

\begin{definition}[Kinetic function]  Let $(X , \mu)$ be a finite measure space. We say that a measurable function $f_+ : X \times \mR \to [0,1]$ is a kinetic function if there exists a Young measure $\nu$ on $X$ such that for every $p \geq 1$,
\begin{gather}
 \int_X \int_\mR \lab \xi \rab^p d \nu_z (\xi) d \mu (z) < + \infty \label{Young}
\end{gather} and for $\mu$-a.e. $z\in X $, for all $\xi \in \mR$,
\begin{gather*}
 f(z , \xi ) = \nu_z (\xi, + \infty).
\end{gather*} \hh \\
Here we recall that a Young measure $\nu$ on $X$ is a weak measurable mapping $z \mapsto \nu_z$ from $X$ into the space of probability measures on $\mR$. For a kinetic function $f_+ : X \times \mR \to [0,1]$ we denote the conjugate function by $f_- = f_+ -1$. Observe that if $f_+ = {\bf 1}_{u > \xi}$, then it is a kinetic function with the corresponding Young measure $\nu = - \de_{u = \xi}$, the Dirac measure centered at $u$, and its conjugate $f_- = - {\bf 1}_{u \leq \xi}$.\\
\indent We introduce the definition of generalized kinetic solution.
\end{definition}

\begin{definition}[Generalized kinetic solution] Let $u_0$ and $u_b$ satisfy ${\rm (H_3)}$. A measurable function $f_+ : \Omega \times Q \times \mR \to [0,1]$ is said to be a generalized kinetic solution of (\ref{SSCL1})-(\ref{SSCL3}) if the following conditions (i)-(iii) hold:
 \begin{enumerate}
  \renewcommand{\labelenumi}{(\roman{enumi})}
  \item $\{ f_+(t) \}$ is predictable. 
  \item $f_+$ is a kinetic function with the associated Young measure $\nu$ on $\Omega \times Q$ such that for all $p \geq 1$, there exists $C_p \geq 0$ satisfying that for a.e. $t \in [0,T]$,
   \begin{gather}
    \mE \int_D \int_{\mR} \lab \xi \rab^p d\nu_{t,x} (\xi) dx \leq C_p. \label{esti26}
   \end{gather}
  \item There exists a kinetic measure $m$ and, for any $N>0$, there exist nonnegative $\bar{m}^\pm_N \in L^1 (\Omega \times \Si \times (-N,N) )$ such that $\{ \bar{m}^\pm_N (t) \}$ are predictable, $\bar{m}^+_N (N-0) = \bar{m}^-_N (-N+0) = 0$ for sufficiently large $N$ and for all $\ph \in C_c^{\infty} ([0,T) \times \ov{D} \times (-N,N))$,
\begin{eqnarray}
    & & \hs{-11mm} \int_{Q \times \mR} f_\pm (\pa_t + a(\xi) \cdot \nabla) \ph \hs{0.5mm} d\xi dx dt \nonumber \\
    & &  + \int_{D \times \mR} f^0_\pm \ph(0) \hs{0.5mm} d\xi dx + M_N \int_{\Si \times \mR} f^b_\pm \ph \hs{0.5mm} d\xi d\si (x) dt \nonumber \\
    & & \hs{-11mm} = - \sum_{k = 1}^\infty \int_0^T \int_{D \times \mR} g_k (x,\xi) \hh \ph (t,x,\xi) \hh d \nu_{t,x} (\xi) dx d\beta_k(t) \nonumber \\
    & &  -\frac{1}{2} \int_{Q \times \mR} G^2 (x,\xi) \hh  \pa_{\xi} \ph (t,x,\xi)  \hh d \nu_{t,x} (\xi) dx dt \nonumber \\
    & &  +\int_{[0,T) \times D \times \mR} \pa_{\xi} \ph \hs{0.5mm} dm + \int_{\Si \times \mR} \pa_{\xi} \ph \hh \bar{m}^\pm_N \hh d\xi d\si (x) dt  \hs{5mm} {\rm a.s.} \label{GKIFO}
   \end{eqnarray}
 \end{enumerate}
\end{definition}

\begin{remark} In the case that the boundary function $\bar{m}^+$ satisfies $\bar{m}^+(-N+0)=0$ in addition, the equality (\ref{GKIFO}) for $f_-$ follows from that for $f_+$ if we set $\bar{m}^- (\xi) = \bar{m}^+ (\xi) + M_N(\xi + N) + (A (\xi ) - A (-N)) \cdot {\bf n} (x)$. In the case of periodic boundary condition as in \cite{a} the boundary function $\bar{m}^+$ does not appear. Thus, in these cases, it is enough to consider the equality (\ref{GKIFO}) only for $f_+$ in the definition of generalized kinetic solutions.
\end{remark}

The following proposition due to \cite[Proposition 8]{a} shows that any generalized kinetic solution admits left and right limits at every $t \in [0,T]$.\\

\begin{lemma} \label{abcdefg} {\it Let $f_+$ be a generalized kinetic solution of }(\ref{SSCL1})-(\ref{SSCL3}){\it. Then $f_+$ admits almost surely left and right limits at all points $t^* \in [0,T]$ in the following sense: For all $t^* \in [0,T]$ there exist some kinetic functions $f_+^{*,\pm}$ on $\Omega \times D \times \mR$ such that $\mP$-a.s., }
\begin{gather*}
 \int_{D \times \mR} f_+ (t^* \pm \ep) \ph \hh d\xi dx \to \int_{D \times \mR} f_+^{*, \pm} \ph \hh d\xi dx
\end{gather*}{\it as $\ep \to +0$ for all $\ph \in C_c^1 (D \times \mR) $. Moreover, almost surely, $f_+^{*,+} = f_+^{*,-}$ for all $t^* \in [0,T]$ except some countable set.}
\end{lemma}

In what follows, for a generalized kinetic solution $f_+$, we will define $f_+^\pm$ by $f_+^\pm (t^*) = f_+^{*, \pm}$ for $t^* \in [0,T]$.\\
\indent In order to prove uniqueness we need to extend test functions in (\ref{GKIFO}) to the class of $C_c^{\infty} ([0,T) \times \mR^d \times \mR)$. To this end we introduce the cutoff functions as follows.
\begin{eqnarray*}
& & \Psi_\eta^+ (\xi) = \int_0^{N-\xi} \psi_\eta (r-\eta) dr, \hs{4mm} \Psi_\eta^- (\xi ) = \int_0^{\xi + N} \psi_\eta (r-\eta) dr \\
& & {\rm and} \hs{4mm} \Psi_\eta (\xi) = \Psi_\eta^+ (\xi ) \Psi_\eta^- (\xi) , \hs{4mm} \eta >  0.\\
\end{eqnarray*}

\begin{proposition} \label{ggghhh} {\it Let $f_+$ be a generalized kinetic solution of }(\ref{SSCL1})-(\ref{SSCL3}){\it . Let $\bar{f}^{(\lam)}_\pm$ be any weak* limit of $\{ f^{\lam, \ep}_\pm \}$ as $\ep \to + 0$ in $L^\infty (\Si^\lam \times \mR)$ for any element $\lam$ of the partition of unity $\{ \lam_i \}$ on $\ov{D}$, where $f_{\pm}^{\lam , \ep}$ is defined by}
\begin{gather*}
 f_\pm^{\lam , \ep} (t,x,\xi) = \int_{D^\lam} f_\pm (t,y,\xi) \rho_\ep^\lam (y-x) dy,
\end{gather*}and let $\bar{f}_{\pm} = \sum_{i=0}^{M} \lam_i \bar{f}_\pm^{(\lam_i)}$.
\begin{enumerate}
  \renewcommand{\labelenumi}{(\roman{enumi})}
  \item {\it For a.s. there exists a full set $\mL$ of $\Si$ such that $\bar{f}_\pm (t,x,\xi)$ is non-increasing in $\xi$ for all $(t,x) \in \mL$.}
  \item {\it For any $\ph \in C_c^\infty (\mR^d \times \mR)$, for any $t \in [0,T)$ and for any $\eta > 0$,}
   \begin{eqnarray}
    & & \hs{-15mm} - \int_{D} \int_{-N}^N \Psi_{\eta} f^{+}_\pm (t) \ph d\xi dx + \int_0^t \int_D \int_{-N}^N  \Psi_{\eta} f_\pm a (\xi) \cdot \nabla \ph d\xi dx ds \nonumber \\
    & & \hs{-10mm} + \int_{D} \int_{-N}^N  \Psi_{\eta} f^0_\pm \ph d\xi dx + \int_0^t \int_D \int_{-N}^N \Psi_{\eta} (-a(\xi) \cdot {\bf n}) \bar{f}_\pm \ph d\xi d\si ds \nonumber \\
    & & \hs{-15mm} = - \sum_{k \geq 1} \int_0^t \int_{D} \int_{-N}^N \Psi_{\eta} g_k \hh \ph \hh d\nu_{x,s}(\xi) dx d\beta_k(s) \nonumber \\
    & & \hs{-10mm} - \frac{1}{2} \int_0^t \int_D \int_{-N}^N \Psi_{\eta} \pa_{\xi} \ph \hh G^2 \hh d\nu_{s,x}(\xi) dx ds + \int_{[0,t] \times D \times (-N,N)} \Psi_{\eta} \pa_{\xi} \ph \hh dm \nonumber \\
    & & \hs{-10mm} + \frac{1}{2} \int_0^t \int_D \int_{-N}^N \Bigl( \psi_{\eta} (N - \xi -\eta) - \psi_\eta (\xi + N - \eta) \Bigr) G^2 \hh \ph \hh d\nu_{s,x} (\xi) dx ds \nonumber \\
    & & \hs{-10mm} - \int_{[0,t] \times D \times (-N,N)} \Bigl( \psi_{\eta} (N - \xi -\eta) - \psi_\eta (\xi + N - \eta) \Bigr) \ph \hh dm \hs{4mm} {\it a.s..}  \label{VKIFO}
   \end{eqnarray}
  \item {\it $P$-a.s., for a.e. $(t,x) \in \Si$, the weak* limits $ \ -a(\xi) \cdot {\bf n} (\bar{x}) \bar{f}_\pm (t,x,\xi)$ coincide with $M_N  f^b_\pm (t,x,\xi) + \pa_{\xi} \bar{m}_N^\pm (t,x,\xi)$ for a.e. $\xi \in (-N,N)$.} \\
 \end{enumerate}
\end{proposition}

\begin{proof} For a.s. let us denote by $\mL_\mR$ the set of Lebesgue points of $\bar{f}_\pm \in L^\infty (\Si \times \mR) $. Take $(t,x,\xi_i) \in \mL_\mR$, $i=1,2$, arbitrarily so that $\xi_1 < \xi_2$. If $\ep , \de > 0$ are sufficiently small, then the average of $f_\pm^{\lam , \ep}$ on $B^i_\de$ satisfy
\begin{gather*}
 \dashint_{B^1_\de} f^{\lam , \ep}_\pm (s,y,\xi) d\xi d \si (y) ds \geq \dashint_{B^2_\de} f_\pm^{\lam , \ep} (s,y,\xi) d\xi d \si (y) ds,
\end{gather*}where $B^i_\de$ denotes the ball with center $(t,x,\xi_i)$ and radius $\de$. Passing to a weak* limit $\bar{f}_\pm^{(\lam)}$ as $\ep_n \to +0$ through some subsequence $\{ f^{\lam , \ep_n}_\pm \}$, we have
\begin{gather*}
 \dashint_{B^1_\de} \bar{f}_\pm^{(\lam)} (s,y,\xi) d\xi d \si (y) ds \geq \dashint_{B^2_\de} \bar{f}_\pm^{(\lam)} (s,y,\xi) d\xi d \si (y) ds.
\end{gather*}Letting $\de \to +0$ and summing over $i$ yield $\bar{f}_\pm (t,x,\xi_1) \geq \bar{f}_\pm (t,x,\xi_2) $ because $(t,x,\xi_i) \in \mL_\mR$. Consequently, setting $\mL = \{ (t,x): (t,x,\xi) \in \mL_\mR \ {\rm for \ a.e.} \ \xi \}$, we obtain the claim of (i).\\
\indent To prove (ii) we take $\phi \in C_c^\infty ([0,T) \times \mR^d \times \mR)$. Putting $\ph = \Psi_\eta \bar{\Theta}_\ep \phi^\lam$ in (\ref{GKIFO}), where $\phi^\lam = \phi \lam$ and
\begin{gather*}
 \bar{\Theta}_\ep (x) = \int^{x_d - h_\lam (\bar{x})}_0 \psi_\ep (r - \ep (L_\lam + 1)) dr,
\end{gather*}we obtain at the limit $\ep \to +0$
 \begin{eqnarray}
    & & \hs{-10mm} \int_{Q^\lam} \int_{-N}^N  \Psi_{\eta} f_\pm (\pa_t + a(\xi) \cdot \nabla) \phi^\lam d\xi dx dt + \int_{D^{\lam}} \int_{-N}^N  \Psi_{\eta} f^0_\pm \phi^\lam(0) d\xi dx \nonumber \\
    & & + \int_0^T \int_{\Pi^\lam} \int_{-N}^N \Psi_\eta (-a(\xi) \cdot {\bf n}_\lam) \bar{f}^{(\lam)}_\pm \phi^\lam d\xi d\bar{\si}_\lam dt \nonumber \\
    & & \hs{-10mm} \hs{2mm} = - \sum_{k \geq 1} \int_0^T \int_{D^\lam} \int_{-N}^N \Psi_{\eta} g_k \hh \phi^\lam \hh d\nu_{x,t}(\xi) dx d\beta_k(t) \nonumber \\
    & & \hs{-10mm} \hs{5.3mm} - \frac{1}{2} \int_{Q^\lam} \int_{-N}^N \Psi_{\eta} \pa_{\xi} \phi^\lam \hh G^2 \hh d\nu_{x,t}(\xi) dx dt + \int_{[0,T) \times D^\lam \times (-N,N)} \Psi_{\eta} \pa_{\xi} \phi^\lam \hh dm \nonumber \\
    & & \hs{-10mm} \hs{5.3mm} + \frac{1}{2} \int_{Q^\lam} \int_{-N}^N \Bigl( \psi_{\eta} (N - \xi -\eta) - \psi_\eta (\xi + N - \eta) \Bigr) G^2 \hh \phi^\lam \hh d\nu_{x,t} (\xi) dx dt \nonumber \\
    & & \hs{-10mm} \hs{5.3mm} - \int_{[0,T) \times D^\lam \times (-N,N)} \Bigl( \psi_{\eta} (N - \xi -\eta) - \psi_\eta (\xi + N - \eta) \Bigr) \phi^\lam \hh dm \hs{4mm} {\rm a.s.},  \label{VKF}
   \end{eqnarray}where
   \begin{gather*}
    {\bf n}_\lam (\bar{x}) = \frac{1}{\sqrt{1+ \lab \nabla_{\bar{x}} h_\lam (\bar{x}) \rab^2 } } (\nabla_{\bar{x}} h_\lam (\bar{x}) , -1), \\
    d\bar{\si}_\lam (\bar{x}) = \sqrt{1+ \lab \nabla_{\bar{x}} h_\lam (\bar{x}) \rab^2} d\bar{x}.
   \end{gather*}In this procedure it will be enough to consider the term
\begin{eqnarray*}
 & & \hs{-2mm} \int_{Q^\lam} \int_{-N}^N \Psi_\eta f_\pm \ph^\lam a (\xi) \cdot \nabla \bar{\Theta}_\ep d\xi dx dt \\
 & & \hs{-2mm} = - \int_0^T \int_{\Pi^\lam} \int_{-N}^N \Psi_\eta a (\xi) \cdot {\bf n}_\lam \int_\mR f_\pm \ph^\lam \rho_\ep (x_d - h_\lam (\bar{x}) - \ep (L_\lam + 1)) dx_d  d\xi d\bar{\si}_\lam dt \\
 & & \hs{-2mm} = - \int_0^T \int_{\Pi^\lam} \int_{-N}^N \Psi_\eta a (\xi) \cdot {\bf n}_\lam \int_{D^\lam} f_\pm (y) \ph^\lam (y) \rho_\ep^\lam (y-x) dy  d\xi d\bar{\si}_\lam dt,
\end{eqnarray*}
which is convergent to
\begin{gather*}
 - \int_0^T \int_{\Pi^\lam} \int_{-N}^N \Psi_\eta a (\xi) \cdot {\bf n}_\lam \bar{f}^{(\lam)}_\pm \phi^\lam d\xi d \bar{\si}_\lam dt,
\end{gather*}with any weak* limit $\bar{f}^{(\lam)}_\pm$ of $\{ f^{\lam , \ep}_\pm \}$ as a corresponding subsequence $\ep_n \to 0$. Let $\ph \in C_c^\infty (\mR^d \times \mR)$. Take a sequence $\{ \ph_n \} \subset C_c^\infty ([0,T) \times \mR^d \times \mR)$ of test function in (\ref{VKF}) which is approximate to ${\bf 1}_{[0,t)}(s) \ph$. By letting $n \to \infty$ and by summing over $i$, we obtain (\ref{VKIFO}) and hence the proof of (ii) is complete. \\
\indent Finally we show (iii). We fix small $\ep > 0$ arbitrarily and take $\ph \in C_c^\infty ([0,T) \times \mR^d \times (-N+\ep , N-\ep))$. Since $(\psi_\eta (N - \xi - \eta) - \psi_\eta (\xi + N - \eta)) \ph = 0$ and $\Psi_\eta \ph = \ph$ for all sufficiently small $\eta > 0$, (\ref{VKF}) deduces
\begin{eqnarray}
    & & \hs{-10mm} \int_{Q} \int_\mR  f_\pm (\pa_t + a(\xi) \cdot \nabla) \ph d\xi dx dt + \int_{D} \int_\mR  f^0_\pm \ph(0) d\xi dx \nonumber \\
    & & + \int_{\Si } \int_\mR  -a(\xi) \cdot {\bf n} \bar{f}_\pm \ph d\xi d\si dt \nonumber \\
    & & \hs{-10mm} \hs{2mm} = - \sum_{k \geq 1} \int_0^T \int_{D } \int_\mR g_k \hh \ph \hh d\nu_{t,x}(\xi) dx d\beta_k(t) \nonumber \\
    & & \hs{-10mm} \hs{5.3mm} - \frac{1}{2} \int_{Q } \int_\mR \pa_{\xi} \ph \hh G^2 \hh d\nu_{t,x}(\xi) dx dt + \int_{[0,T) \times D \times \mR} \pa_{\xi} \ph \hh dm \hs{4mm} {\rm a.s.}. \label{KIFO2}
   \end{eqnarray}It follows from (\ref{GKIFO}) and (\ref{KIFO2}) that for all $\ph \in C_c^\infty ([0,T) \times \mR^d \times (-N+\ep , N-\ep))$,
\begin{eqnarray*}
& & \int_{\Si } \int_\mR - a (\xi) \cdot {\bf n} \bar{f}_\pm \ph  =  M_N \int_{\Si} \int_\mR f^b_\pm \ph  - \int_{\Si} \int_\mR \pa_\xi \ph \bar{m}_N^\pm,
\end{eqnarray*}which implies that
\begin{gather*}
 \pa_\xi \bar{m}^\pm_N = - a (\xi) \cdot {\bf n} \bar{f}_\pm - M_N f^b_\pm \in L^1 (\Si \times (-N+ \ep ,N+\ep ) )
\end{gather*}in the sense of distribution on $\Si \times (-N+\ep , N - \ep)$. By Nikodym's theorem, for a.e. $(t,x) \in \Si$, $\bar{m}^\pm (t,x,\xi)$ is absolutely continuous in $\xi$ and
\begin{gather*}
\pa_\xi \bar{m}^\pm_N (t,x,\xi) = - a (\xi) \cdot {\bf n} (x) \bar{f}_\pm (t,x,\xi) - M_N f^b_\pm (t,x,\xi),
\end{gather*}for a.e. $\xi \in (-N+\ep , N-\ep)$. Since $\ep > 0 $ is arbitrary, we conclude the desired claim. \\
\end{proof}

\section{Uniqueness}

In this section we prove the main result of the paper.

\begin{theorem}[$L^1$-contraction property] \label{contraction} {\it Assume that $D$ is a bounded domain with Lipschitz boundary. Let $f_{i,+}$, $i=1,2$, be generalized kinetic solutions to }(\ref{SSCL1})-(\ref{SSCL3}){\it with data}
 $(f^0_{i,+},f^b_{i,+})=({\bf 1}_{u_{i,0}>\xi} , {\bf 1}_{u_{i,b}>\xi})${\it , respectively. Under the assumptions ${\rm (H_1)}$-${\rm (H_3)}$ we have for a.e. $t \in [0,T)$}
\begin{eqnarray}
& & \hs{-5mm} - \mE \int_D \int_\mR f_{1,+} (t,x,\xi) f_{2,-} (t,x,\xi) \leq -\mE \int_D \int_\mR f_{1,+}^0 (x,\xi) f_{2,-}^0 (x,\xi) \nonumber \\
& & \hs{30mm} - M_b \mE \int^t_0 \int_{\pa D} \int_\mR f_{1,+}^b (s,x,\xi) f_{2,-}^b(s,x,\xi), \label{contraction2}
\end{eqnarray}{\it where $M_b = \max \{ \lab a(\xi) \rab : \lab \xi \rab \leq \lno u_{1,b} \rno_{L^\infty (\Omega \times \Si )} \vee \lno u_{2,b} \rno_{L^\infty (\Omega \times \Si )} \}$.}
\end{theorem}

\begin{corollary}[Uniqueness, Reduction] \label{uniqueness} {\it Under the same assumptions as in the above theorem, if $f_+$ is a generalized solution to }(\ref{SSCL1})-(\ref{SSCL3}) {\it with initial datum ${\bf 1}_{u_0 > \xi}$ and boundary datum ${\bf 1}_{u_b > \xi}$, then there exists a kinetic solution $u$ to }(\ref{SSCL1})-(\ref{SSCL3}) {\it with initial datum $u_0$ and boundary datum $u_b$ such that $f_+ (t,x,\xi) = {\bf 1}_{u(t,x) > \xi}$ a.s. for a.e. $(t,x,\xi)$. Moreover, for a.e. $t \in [0,T)$,}
\begin{eqnarray}
& & \hs{-14mm} \mE \lno u_1 (t) - u_2 (t) \rno_{L^1(D)} \leq \mE \lno u_{1,0} - u_{2,0} \rno_{L^1 (D)} \nonumber \\
& & \hs{35mm} + M_b \mE \lno u_{1,b} - u_{2,b} \rno_{L^1(\Si)},
\end{eqnarray}{\it where $u_i$, $i=1,2$, are the corresponding kinetic solutions to }(\ref{SSCL1})-(\ref{SSCL3}){\it with data $(u_{i,0} , u_{i,b})$.}
\end{corollary}

To prove the uniqueness theorem we define the non-increasing functions $\mu_m (\xi)$ and $\mu_\nu (\xi)$ on $\mR$ by
\begin{gather}
 \mu_m (\xi) = \mE m ([0,T) \times D \times (\xi, \infty)) \hs{4mm} {\rm and} \label{mum} \\
 \mu_\nu (\xi) = \mE \int_{Q \times (\xi , \infty)} d \nu_{t,x} (\xi) dx dt, \label{munu}
\end{gather}where $m$ and $\nu$ are a kinetic measure and a Young measure satisfying (\ref{Young}), respectively. Let $\mD$ be the set of $\xi \in (0,\infty) $ such that both of $\mu_m$ and $\mu_\nu$ are differentiable at $-\xi$ and $\xi$. It is easy to see that $\mD$ is a full set in $(0,\infty)$.

\begin{lemma} \label{limit}
\begin{enumerate}
 \renewcommand{\labelenumi}{(\roman{enumi})}
  \item $\displaystyle \limsup_{\xi \to \infty, \hh \xi \in \mD} \mu_m ' (\pm \xi) = 0 $ \hs{2mm} {\it and \hs{2mm} $\displaystyle \limsup_{\xi \to \infty , \hh \xi \in \mD} \xi^p \mu_\nu ' (\pm \xi) = 0$ for $p \geq 1$. }
  \item {\it If $N \in \mD $, then as $\de \to + 0$}
 \begin{gather*}
  \int_\mR \psi_\de (N \pm \zeta) d \mu_m (\zeta ) \to \mu_m ' (\mp N)
 \end{gather*}{\it and}
 \begin{gather*}
  \int_\mR \psi_\de (N \pm \zeta) (1 + \lab \zeta \rab^2 ) d \mu_\nu (\zeta ) \to (1 + N^2) \mu_\nu ' (\mp N).
 \end{gather*}
 \end{enumerate}
\end{lemma}

\begin{proof} We prove the lemma only in the case of $\mu_\nu$. The case of $\mu_m$ will be done in a similar fashion. Due to (\ref{Young}) there exists $C_p > 0$ such that $\lab \xi \rab^p \mu_\nu (\xi) \leq C_p$ for every $\xi \in \mR$. Let us assume that $\displaystyle \limsup_{\xi \to \infty , \hh \xi \in \mD} \xi^p \mu_\nu ' (\pm \xi) = \alpha < 0 $. Then we can take $\xi_0 \in \mD$ so that $\xi^p \mu_\nu ' (\pm \xi ) < \alpha / 2 $ whenever $\xi > \xi_0$ and $\xi \in \mD$. Since the function $\xi \mapsto \xi^p \mu_\nu (\xi)$ is non-increasing on $(\xi_0 , \infty)$ if $\xi_0$ is sufficiently large, we have
\begin{eqnarray*}
 \xi^p \mu_\nu (\xi) - \xi_0^p \mu_\nu (\xi_0) & \leq & \int_{\xi_0}^{\xi} (\zeta^p \mu_\nu (\zeta)) ' d \zeta \\
 & \leq & \int_{\xi_0}^{\xi} \la \frac{C_p}{\zeta} + \frac{\alpha}{2} \ra d \zeta.
\end{eqnarray*}Hence $\displaystyle \limsup_{\xi \to \infty} \xi^p \mu_\nu (\xi) = - \infty$ and this contradicts the fact that $\mu_\nu (\xi ) \geq 0$. On the other hand the function $\xi \mapsto \xi^p \mu_\nu (- \xi)$ is non-decreasing on $(\xi_0 , \infty)$. Hence
\begin{eqnarray*}
\xi^p \mu_\nu (-\xi) - \xi_0^p \mu_\nu (-\xi_0) & \geq & \int_{\xi_0}^{\xi} \la \zeta^p \mu_\nu (- \zeta) \ra ' d \zeta \\
& \geq & \int_{\xi_0}^{\xi} \la - \frac{\alpha}{2} \ra d \zeta.
\end{eqnarray*}Therefore, $\displaystyle \liminf_{\xi \to \infty} \xi^p \mu_\nu (-\xi) = \infty$ and this contradicts the fact that $\lab \xi \rab^p \mu_\nu (\xi) \leq C_p$. Consequently we have $\displaystyle \limsup_{\xi \to \infty, \hh \xi \in \mD} \xi^p \mu_\nu ' (\pm \xi) = 0$. \\
\indent Next, let $N \in \mD$. Since $\mu_\nu (\pm (N - \zeta)) = \mu_\nu (\pm N) + \mu_\nu ' (\pm N) \zeta + o (\zeta)$, it follows that
\begin{eqnarray*}
& & \int_\mR \psi_\de (N \pm \zeta ) (1 + \lab \zeta \rab^2) d \mu_\nu (\zeta) \\
& & \hs{-2mm} = - \int_{-\de}^\de \mu_\nu (\mp N + \zeta) d (\psi_\de (\zeta) (1 + ( \mp N + \zeta)^2 ) ) \\
& & \hs{-2mm} =  \mu_{\nu}' (\mp N) \left\{ (1+ N^2) + \int_{-\de}^\de \zeta^2 \psi_\de (\zeta ) d\zeta \right\} + \int_{-\de }^\de o(\zeta) d (\psi_\de (\zeta )( 1 + (\mp N + \zeta)^2 )).
\end{eqnarray*}Besides, the last term of the right hand on the above equality tends to $0$ as $\de \to 0$. To see this take an arbitrary $\ep > 0$. There exists $\de_0 > 0$ such that $\lab o(\zeta) \rab \leq \ep \lab \zeta \rab$. If $0 < \de < \de_0$, then
\begin{eqnarray*}
& & \lab \int_{-\de}^\de o(\zeta) d (\psi_\de (\zeta)(1+ (\mp N + \zeta)^2)) \rab \\
& & \leq \ep \int_{-\de }^\de \lab \zeta \psi_\de ' (\zeta)(1+ (\mp N + \zeta)^2) + 2 (\mp N + \zeta) \psi_\de (\zeta) \rab d \zeta \\
& & \leq \ep \left\{ \de (1 + (\pm N \pm \de )^2 ) + 2 \lab \mp N \mp \de \rab \right\}.
\end{eqnarray*}Thus we obtain the claim of (ii) for $\mu_\nu$.
\end{proof}

\begin{proposition}[Doubling variable] \label{Doubling} {\it Let $f_{i,+}$, $i=1,2$, be generalized kinetic solutions to }(\ref{SSCL1})-(\ref{SSCL3}){\it with data $(f_{i,+}^0 , f_{i,+}^b )$. Then, for $t \in [0,T)$, for $\ep , \de >0$, for $N \in \mD$ and for any element $\lam$ of the partition of unity $\{ \lam_i \}$ on $\ov{D}$, we have}
\begin{eqnarray}
& & \hs{-3mm} - \mE \int_{D^2 \times (-N,N)^2} \lam (x) \rho_\ep^\lam (y-x) \psi_\de (\xi - \zeta) f_{1,+}^+ (t,x,\xi) f_{2,-}^+ (t,y,\zeta) d\xi d\zeta dxdy \nn \\
& & \hs{-3mm} \leq - \mE \int_{D^2 \times (-N,N)^2} \lam (x) \rho_\ep^\lam (y-x) \psi_\de (\xi - \zeta) f_{1,+}^0 (x,\xi) f_{2,-}^0 (y,\zeta) d\xi d\zeta dxdy \nn \\
& & \hs{-1mm} - \mE \int_{(0,t) \times \pa D \times D \times (-N,N)^2} \lam (x) \rho_\ep^\lam (y-x) \psi_\de (\xi - \zeta) (-a(\xi) \cdot {\bf n} (x)) \nn \\
& & \hs{49mm} \times \bar{f}_{1,+}^{(\lam)} (s,x,\xi) f_{2,-} (s,y,\zeta) d\xi d\zeta d\si (x) dy ds \nn \\
& & + I_1 + I_2 + I_3 + I_N, \label{doubling}
\end{eqnarray}where
\begin{eqnarray*}
& & \hs{-2mm}  I_1 = - \mE \int_{(0,t) \times D^2 \times (-N,N)^2}  f_{1,+} (s,x,\xi) f_{2,-} (s,y,\zeta) (a(\xi) - a(\zeta)) \\
& & \hs{48mm} \cdot  \nabla_x  \rho_\ep^\lam (y-x)  \lam (x) \psi_\de (\xi - \zeta) d\xi d\zeta dx dy ds , \\
& & \hs{-2mm} I_2 = - \mE \int_{(0,t) \times D^2 \times (-N,N)^2} f_{1,+} (s,x,\xi) f_{2,-} (s,y,\zeta) a (\xi) \\
& & \hs{48mm} \cdot \nabla_x \lam (x) \rho^\lam_\ep (y-x) \psi_\de (\xi - \zeta) d\xi d\zeta dx dy ds, \\
& & \hs{-2mm} I_3 = \frac{1}{2} \mE \int_{(0,t) \times D^2 \times (-N,N)^2} \lam (x) \rho_\ep^\lam (y-x) \psi_\de (\xi - \zeta) \\
& & \hs{28mm} \times \sum_{k=1}^\infty \lab g_k (x,\xi) - g_k (y , \zeta) \rab^2 d \nu_{s,x}^1 (\xi) \otimes d\nu_{s,y}^2 (\zeta) dxdyds, \\
& & \hs{-2mm} \limsup_{N \to \infty} I_N = 0 \hs{4mm} {\rm with} \ I_N \ {\rm defined \ by \ (\ref{IN}) \ below}.
\end{eqnarray*}Here $m^i$ and $\nu^i$, $i=1,2$, are the associated kinetic measures and the associated Young measures with the generalized kinetic solutions $f_{i,+}$, $\bar{f}_{i,\pm}^{(\lam)}$ are any weak* limits of $\{ f^{\lam ,\ep'}_{i,\pm} \}$ as $\ep '  \to 0$ in $L^\infty (\Si^\lam \times \mR) $, and $C$ is a constant which is independent of $\ep$, $\de$, $N$.
\end{proposition}

\begin{proof} We will follow the proof of \cite[Proposition 9]{a}. Let $\ph_1 \in C_c^\infty (\mR^d_x \times \mR_\xi)$ and $\ph_2 \in C_c^\infty (\mR^d_y \times \mR_\zeta)$. Set
\begin{gather*}
 F_{1,+} (t) = \sum_{k=1}^\infty \int_0^t \int_{D_x^\lam} \int_{-N}^N \Psi_\eta (\xi) g_{k,1} \ph^\lam_1 d\nu_{s,x}^1 (\xi) dx d\beta_k (s) 
\end{gather*}and
\begin{eqnarray*}
& & \hs{-6mm} G_{1,+} (t) =  \int^t_0 \int_{D_x^\lam} \int_{-N}^N \Psi_\eta (\xi) f_{1,+} (s,x,\xi) a (\xi) \cdot \nabla_x \ph^\lam_1 d\xi dx ds \\
& & +  \frac{1}{2}   \int^t_0 \int_{D_x^\lam} \int_{-N}^N \Psi_\eta (\xi) \pa_\xi \ph^\lam_1 G_1^2 d\nu_{s,x}^1(\xi) dx ds \\
& &  +   \int^t_0 \int_{\pa D_x^\lam} \int_{-N}^N \Psi_\eta (\xi) (-a(\xi) \cdot {\bf n}) \bar{f}^{(\lam)}_{1,+} (s,x,\xi) \ph^\lam_1 d\xi d\si (x) ds \\
& & -   \int_{[0,t] \times D^\lam_x \times  (-N,N)} \Psi_\eta (\xi) \pa_\xi \ph^\lam_1 dm^1 (s,x,\xi) \\
& &  -  \frac{1}{2}  \int^t_0 \int_{D_x^\lam} \int_{-N}^N ( \psi_\eta (- \eta +N - \xi) - \psi (-\eta + N + \xi)) \ph^\lam_1 G_1^2 d\nu_{s,x}^1 (\xi) dx ds \\
& &  +  \int_{[0,t] \times D^\lam_x \times (-N,N)} ( \psi_\eta (- \eta +N - \xi) - \psi (-\eta + N + \xi)) \ph^\lam_1 dm^1 (s,x,\xi).
\end{eqnarray*}On the other hand we set
\begin{eqnarray*}
& & \hs{-6mm} F_{2,-} (t) = \sum_{k=1}^\infty \int_0^t \int_{D_y} \int_{-N}^N \Psi_\eta (\zeta) g_{k,2} \ph_2 d\nu_{s,y}^2 (\zeta) dy d\beta_k (s) , \\
& & \hs{-6mm} G_{2,-} (t)  =  \int^t_0 \int_{D_y} \int_{-N}^N \Psi_\eta (\zeta) f_{2,-} (s,x,\zeta) a (\zeta) \cdot \nabla_y \ph_2 d\zeta dy ds \\
& &  +  \frac{1}{2}   \int^t_0 \int_{D_y} \int_{-N}^N \Psi_\eta (\zeta) \pa_\zeta \ph_2 G_2^2 d\nu_{s,y}^2(\zeta) dy ds \\
& &  +   \int^t_0 \int_{\pa D_y} \int_{-N}^N \Psi_\eta (\zeta) (-a(\zeta) \cdot {\bf n (y)}) \bar{f}_{2,-} (s,y,\zeta) \ph_2 d\zeta d\si (y) ds \\
& &  -   \int_{[0,t] \times D_y \times (-N,N)} \Psi_\eta (\zeta) \pa_\zeta \ph_2 dm^2 (s,y,\zeta) \\
& & -  \frac{1}{2}  \int^t_0 \int_{D_y} \int_{-N}^N ( \psi_\eta (- \eta +N - \zeta) - \psi (-\eta + N + \zeta)) \ph_2 G_2^2 d\nu_{s,y}^2 (\zeta) dy ds \\
& &  +  \int_{[0,t] \times D_y \times (-N,N)} ( \psi_\eta (- \eta +N - \zeta) - \psi (-\eta + N + \zeta)) \ph_2 dm^2 (s,y,\zeta).
\end{eqnarray*}
By (\ref{VKIFO}) and (\ref{VKF}) we have
\begin{eqnarray*}
& & \int_{D_x^\lam} \int_{-N}^N \Psi_\eta (\xi) f_{1,+}^+ (t) \ph_1^\lam = F_{1,+} (t) + G_{1,+} (t) + \int_{D_x^\lam} \int_{-N}^N \Psi_\eta (\xi) f_{1,+}^0 \ph_1^\lam
\end{eqnarray*}and
\begin{eqnarray*}
& & \int_{D_y} \int_{-N}^N \Psi_\eta (\zeta) f_{2,-}^+ (t) \ph_2 = F_{2,-} (t) + G_{2,-} (t) + \int_{D_y} \int_{-N}^N \Psi_\eta (\zeta) f_{2,-}^0 \ph_2 .
\end{eqnarray*}Set $\alpha (x,\xi,y,\zeta) = \ph_1 (x,\xi) \ph_2 (y,\zeta)$ and $\Psi_\eta (\xi,\zeta) = \Psi_\eta (\xi) \Psi_\eta (\zeta)$. Using It${\rm \hat{o}}$'s formula for $F_{1,+} (t) F_{2,-}(t)$, integration by parts for functions of finite variation (see \cite[p.6]{q}) for 
\begin{gather*}
 \bigg\{ G_{1,+}(t) + \int_{D_x^\lam} \int_{-N}^N \Psi_\eta (\xi) f_{1,+}^0 \ph_1^\lam  \bigg\} \bigg\{   G_{2,-}(t) + \int_{D_y} \int_{-N}^N \Psi_\eta (\zeta) f_{2,-}^0 \ph_2 \bigg\},
\end{gather*}and integration by parts for functions of finite variation and continuous martingales (see \cite[p.152]{q}) for
\begin{gather*}
F_{1,+} (t) \left\{ G_{2,-} (t) + \int_{D_y} \int_{-N}^N \Psi_\eta (\zeta) f_{2,-}^0 \ph_2 d\zeta dy \right\},
\end{gather*}we obtain
\begin{eqnarray}
& & \hs{-5mm} - \mE \int_{D_x^\lam} \int_{D_y} \int_{-N}^N \int_{-N}^N \Psi_\eta (\xi,\zeta) f_{1,+}^+ (t) f_{2,-}^+ (t) \alpha^\lam d\xi d\zeta dx dy \nonumber \\
& & \hs{-4mm} = - \mE \int_{D_x^\lam} \int_{D_y} \int_{-N}^N \int_{-N}^N \Psi_\eta (\xi,\zeta) f_{1,+}^0 f_{2,-}^0 \alpha^\lam  d\xi d\zeta dx dy \nonumber \\
& & \hs{-4mm} - \sum_{k=1}^\infty \mE \int_0^t \int_{D_x^\lam} \int_{D_y} \int_{-N}^N \int_{-N}^N \Psi_\eta (\xi,\zeta) g_{k,1} g_{k,2} \alpha^\lam d\nu_{s,x}^1 (\xi) \otimes d \nu_{s,y}^2 (\zeta) dx dy ds \nonumber \\
& & \hs{-4mm} - \mE \int_0^t \int_{D_x^\lam} \int_{D_y} \int_{-N}^N \int_{-N}^N \Psi_\eta (\xi,\zeta) f_{1,+} (s) f_{2,-} (s) (a(\xi ) \cdot \nabla_x + a (\zeta) \cdot \nabla_y ) \nonumber \\
& & \hs{80mm} \times \alpha^\lam d\xi d\zeta dx dy ds \nonumber \\
& & \hs{-4mm}  -\frac{1}{2}  \mE \int_0^t \int_{D_x^\lam} \int_{D_y} \int_{-N}^N \int_{-N}^N \Psi_\eta (\xi,\zeta) f_{1,+} (s) \pa_\zeta\alpha^\lam G_2^2 d\nu_{s,y}^2 (\zeta) d\xi dx dy ds \nonumber \\
& & \hs{-4mm}  - \mE \int_0^t \int_{D_x^\lam} \int_{\pa D_y} \int_{-N}^N \int_{-N}^N \Psi_\eta (\xi,\zeta) f_{1,+} (s) \bar{f}^{(\lam)}_{2,-} (s) (-a(\zeta) \cdot {\bf n}) \nonumber \\
& & \hs{72mm} \times \alpha^\lam d\xi d\zeta dx d\si (y) ds \nonumber \\
& & \hs{-4mm}  + \mE \int_{[0,t] \times D_y \times (-N,N)} \int_{D_x^\lam} \int_{-N}^N \Psi_\eta (\xi,\zeta) f_{1,+}^- (s) \pa_\zeta \alpha^\lam d\xi dx dm^2 (s,y,\zeta) \nonumber \\
& & \hs{-4mm}  + \frac{1}{2} \mE \int_0^t \int_{D_x^\lam} \int_{D_y} \int_{-N}^N \int_{-N}^N \Psi_\eta (\xi) f_{1,+} (s) \Big[\psi_\eta (-\eta + N - \zeta) \nonumber \\
& & \hs{42mm} - \psi_\eta (-\eta + N + \zeta)\Big] G_2^2 \alpha^\lam d \nu_{s,y}^2 (\zeta) d\xi dx dy ds \nonumber \\
& & \hs{-4mm}  - \mE \int_{[0,t] \times D_y \times (-N,N)} \int_{D_x^\lam} \int_{-N}^N \Psi_\eta (\xi) f_{1,+}^+ (s) \Big[\psi_\eta (-\eta + N - \zeta) \nonumber \\
& & \hs{48mm} - \psi_\eta (-\eta + N + \zeta)\Big] \alpha^\lam d\xi dx dm^2 (s,y,\zeta) \nonumber \\
& & \hs{-4mm}  - \frac{1}{2}  \mE \int_0^t \int_{D_x^\lam} \int_{D_y} \int_{-N}^N \int_{-N}^N \Psi_\eta (\xi,\zeta) f_{2,-} (s) \pa_\xi \alpha^\lam G_1^2 d\nu_{s,x}^1 (\xi) d\zeta dx dy ds \nonumber \\
& & \hs{-4mm}  - \mE \int_0^t \int_{\pa D_x^\lam} \int_{ D_y} \int_{-N}^N \int_{-N}^N \Psi_\eta (\xi,\zeta) \bar{f}^{(\lam)}_{1,+} (s) f_{2,-} (s) (-a(\xi) \cdot {\bf n}) \nonumber \\
& & \hs{-4mm}  \hs{85mm} \times \alpha^\lam d\xi d\zeta d\si (x) dy ds \nonumber \\
& & \hs{-4mm}  + \mE \int_{[0,t] \times D_x^\lam \times (-N,N) } \int_{D_y} \int_{-N}^N \Psi_\eta (\xi,\zeta) f_{2,-}^+ (s) \pa_\xi \alpha^\lam d\zeta dy dm^1 (s,x,\xi) \nonumber \\
& & \hs{-4mm}  + \frac{1}{2} \mE \int_0^t \int_{D_x^\lam} \int_{D_y} \int_{-N}^N \int_{-N}^N \Psi_\eta (\zeta) f_{2,-} (s) \Big[\psi_\eta (-\eta + N - \xi) \nonumber \\
& & \hs{-4mm}  \hs{49mm} - \psi_\eta (-\eta + N + \xi)\Big] G_1^2 \alpha^\lam d \nu_{s,x}^1 (\xi) d\zeta dx dy ds \nonumber \\
& & \hs{-4mm}  - \mE \int_{[0,t] \times D_x^\lam \times (-N,N) } \int_{D_y} \int_{-N}^N \Psi_\eta (\zeta) f_{2,-}^- (s) \Big[\psi_\eta (-\eta + N - \xi) \nonumber \\
& & \hs{-4mm}  \hs{55mm} - \psi_\eta (-\eta + N + \xi)\Big] \alpha^\lam d\zeta dy dm^1 (s,x,\xi) \nonumber \\ \label{nagaisiki}
\end{eqnarray}where $\alpha^\lam = \alpha (x,\xi,y,\zeta) \lam (x) $. Noting that $C_c^\infty (\mR^d_x \times \mR_\xi) \otimes C_c^\infty (\mR^d_y \times \mR_\zeta)$ is dense in $C_c^\infty (\mR_x^d \times \mR_\xi \times \mR_y^d \times \mR_\zeta)$ and that $m^i$ and $\nu^i$, $i=1,2$, vanish for large $\xi$ thanks to (\ref{vanish}) and (\ref{Young}), by an approximation argument we can take $\alpha(x,\xi,y,\zeta) = \rho_\ep^\lam (y-x) \psi_\de (\xi - \zeta)$ in (\ref{nagaisiki}). In this case note that $\alpha^\lam = \lam (x) \rho_\ep^\lam (y-x) \psi_\de (\xi - \zeta)$ and $\rho_\ep^\lam (y-x) = 0$ on $D_x^\lam \times  \pa D_y$. Using the identity $(\pa_\xi + \pa_\zeta) \psi_\de =0$, we compute the fourth and sixth terms on the right hand of (\ref{nagaisiki}) as follows.
\begin{eqnarray*}
& & \hs{-6mm} -\frac{1}{2} \mE \int_0^t \int_{D_x^\lam} \int_{D_y} \int_{-N}^N \int_{-N}^N \Psi_\eta (\xi , \zeta) f_{1,+} (s) \pa_\zeta \alpha^\lam G_2^2 d\nu^2_{s,y} (\zeta) d\xi dx dy ds \\
& & \hs{-6mm} = \frac{1}{2} \mE \int_0^t \int_{D_x^\lam} \int_{D_y} \int_{-N}^N \int_{-N}^N \Psi_\eta (\xi , \zeta) f_{1,+} (s) \pa_\xi \alpha^\lam G_2^2 d\nu^2_{s,y} (\zeta) d\xi dx dy ds \\
& & \hs{-6mm} = -\frac{1}{2} \mE \int_0^t \int_{D_x^\lam} \int_{D_y} \int_{-N}^N \int_{-N}^N \Psi_\eta (\zeta) \Big[ \psi_\eta (N-\eta - \xi) - \psi_\eta (N - \eta + \xi) \Big] \\
& & \hs{50mm}  \times f_{1,+} (s) \alpha^\lam G_2^2 d\nu^2_{s,y} (\zeta) d\xi dx dy ds \\
& & +\frac{1}{2} \mE \int_0^t \int_{D_x^\lam} \int_{D_y} \int_{-N}^N \int_{-N}^N \Psi (\xi , \zeta) \alpha^\lam G_2^2 d\nu_{s,x}^1 (\xi) d\nu^2_{s,y} (\zeta) dx dy ds \\
\end{eqnarray*}and
\begin{eqnarray*}
& & \hs{-6mm} \mE \int_0^t \int_{D_x^\lam} \int_{D_y} \int_{-N}^N \int_{-N}^N \psi_\eta (\xi , \zeta) f_{1,+}^- (s) \pa_\zeta \alpha^\lam d\xi dx dm^2 (s,y,\zeta) \\
& & \hs{-6mm} = - \mE \int_0^t \int_{D_x^\lam} \int_{D_y} \int_{-N}^N \int_{-N}^N \Psi_\eta (\zeta) \Big[ \psi_\eta (N-\eta - \xi) \\
& & \hs{24mm} - \psi_\eta (N-\eta +\xi) \Big] f_{1,+}^- (s) \alpha^\lam d\xi dx dm^2 (s,y,\zeta) \\
& & - \mE \int_0^t \int_{D_x^\lam} \int_{D_y} \int_{-N}^N \int_{-N}^N \Psi_\eta (\xi , \zeta) \alpha^\lam  d\nu_{s,x}^{1,-} (\xi) dx dm^2 (s,y,\zeta) \\
& & \hs{-6mm} \leq - \mE \int_0^t \int_{D_x^\lam} \int_{D_y} \int_{-N}^N \int_{-N}^N \Psi_\eta (\zeta) \Big[ \psi_\eta (N-\eta - \xi) \\
& & \hs{24mm} - \psi_\eta (N-\eta +\xi) \Big] f_{1,+}^- (s) \alpha^\lam d\xi dx dm^2 (s,y,\zeta).
\end{eqnarray*}Similarly, the ninth and eleventh terms can be computed. We then calculate the terms produced by the truncation function $\Psi_\eta$, namely, the terms containing the functions $\psi_\eta (-\eta + N \pm\xi)$ or $\psi_\eta (-\eta + N \pm \zeta)$.
\begin{eqnarray*}
& & \frac{1}{2} \mE \int_0^t \int_{D_x^\lam} \int_{D_y} \int_{-N}^N \int_{-N}^N \Psi_\eta (\xi) f_{1,+} (s) \psi_\eta (-\eta + N \pm \zeta) \\
& & \hs{65mm} \times G_2^2 \alpha^\lam d \nu_{s,y}^2 (\zeta) d\xi dx dy ds \\
& & \leq C \mE \int_0^t \int_{D_x^\lam} \int_{D_y} \int_{-N}^N \int_{-N}^N \psi_\eta (-\eta + N \pm \zeta) \la 1 + \lab \zeta \rab^2 \ra \\
& & \hs{45mm} \times \rho_\ep^\lam (y-x) \psi_\de (\xi - \zeta) d \nu_{s,y}^2 (\zeta) d\xi dx dy ds \\
& &  \leq C \int_{\mR} \psi_\eta (-\eta + N \pm \zeta) \la 1 + \lab \zeta \rab^2 \ra \mE \int^T_0 \int_D d \nu_{s,y}^2 (\zeta) dy ds \\
& & = C \int_{\mR} \psi_\eta (-\eta + N \pm \zeta) \la 1 + \lab \zeta \rab^2 \ra d\mu_{\nu^2} (\zeta) \to \mp C (1+N^2) \mu_{\nu^2} ' ( \pm N)
\end{eqnarray*}as $\eta \to +0$ by virtue of Lemma \ref{limit}, where $\mu_{\nu^2}$ is defined by (\ref{munu}). A similar argument yields that all the other terms containing the function $\psi_\eta$ on the right hand of (\ref{nagaisiki}) are estimated from above as $\eta \to +0$ by
\begin{gather}
 I_N = C \la \mu_{m^1} ' (N) + \mu_{m^2} ' (N) + (1+N^2) (\mu_{\nu^1} ' (N) + \mu_{\nu^2} ' (N))  \ra , \label{IN}
\end{gather}which is convergent to $0$ as $N \to \infty$ by Lemma \ref{limit}. Consequently, letting $\eta \to +0$ in (\ref{nagaisiki}) and then using the identity $(\nabla_x + \nabla_y ) \rho_\ep^\lam =0$ in the third term on the right hand we obtain (\ref{doubling}) with $I_N$ defined by (\ref{IN}).
\end{proof}

\noindent {\it Proof of Theorem \ref{contraction}} \ Set for $t \geq 0$ and $N > 0 $,
\begin{eqnarray*}
& & \eta^t_N (\ep , \de) = - \mE \int_{D_x^\lam} \int_{D_y} \int_{-N}^N \int_{-N}^N \lam (x) \rho_\ep^\lam (y-x) \psi_\de (\xi - \zeta) \\
& & \hs{50mm} \times f_{1,+} (t,x,\xi) f_{2,-} (t,y,\zeta) d\xi d\zeta dx dy \\
& & \hs{15mm} + \mE \int_{D^\lam} \int_{-N}^N \lam (x) f_{1,+} (t, x ,\xi ) f_{2,-} (t,x,\xi) d\xi dx.
\end{eqnarray*}It is easy to see that $\lim_{\ep , \de \to 0} \eta_N^t (\ep , \de) = 0$ uniformly in $N$. Also set
\begin{eqnarray*}
& & \hs{-4mm} r_N (\ep , \de) = - \mE \int_0^t \int_{\pa D_x^\lam} \int_{D_y} \int_{-N}^N \int_{-N}^N \lam (x) \rho_\ep^\lam (y-x) \psi_\de (\xi - \zeta) \\
& &  \hs{16mm} \times (-a(\xi) \cdot {\bf n} (x) ) \bar{f}^{(\lam)}_{1,+} (s,x,\xi) f_{2,-} (s,y,\zeta) d\xi d\zeta d\si (x) dy ds \\
& &  \hs{1mm} + \mE \int_0^t \int_{\pa D_x^\lam} \int_{-N}^N \lam (x) (-a (\xi) \cdot {\bf n} (x) ) \bar{f}^{(\lam)}_{1,+} (s,x,\xi) \bar{f}_{2,-}^{(\lam)} (s,x,\xi) d\xi d\si (x) ds.
\end{eqnarray*}Since there exists a sequence $\{ \ep_n \} \downarrow 0$ such that $f_{2,-} * \rho_{\ep_n}^\lam$ converges as $n \to \infty$ to $\bar{f}^{(\lam)}_{2,-}$ in $L^\infty (\Si^\lam \times \mR)$-weak*, we see that $\lim_{\ep_n, \de \to 0} r_N (\ep_n , \de) = 0$ for each $N > 0$. Therefore, it follows from Proposition \ref{Doubling} that 
\begin{eqnarray*}
& & \hs{-6mm} -\mE \int_{D^\lam} \int_{-N}^N \lam (x) f_{1,+}^+ (t , x,\xi ) f_{2,-}^+ (t , x, \xi) d\xi dx \\
& & \hs{-6mm} \leq  -\mE \int_{D^\lam} \int_{-N}^N \lam (x) f_{1,+}^0 ( x,\xi ) f_{2,-}^0 ( x, \xi) d\xi dx \\
& & -\mE \int_0^t \int_{\pa D^\lam} \int_{-N}^N \lam (x) (-a (\xi) \cdot {\bf n} (x)) \bar{f}^{(\lam)}_{1,+} (s, x,\xi ) \bar{f}^{(\lam)}_{2,-} (s, x, \xi) d\xi d\si (x) ds \\
& & + I_1 + I_2 + I_3 + I_N + \eta_N^t (\ep_n , \de) + \eta_N^0 (\ep_n , \de) + r_N (\ep_n , \de).
\end{eqnarray*}On the domain $U_{\lam_0}$ a similar argument also deduces the same inequality as above, but the term on the boundary $\pa D^{\lam_0}$ vanishes. By virtue of Lemma \ref{ggghhh} (iii) it holds that $a(\xi) \cdot {\bf n} (x) \bar{f}_{2,-}^{(\lam)} = a (\xi) \cdot {\bf n} (x) \bar{f}_{2,-} $ a.e. on $[0,T) \times \pa D^\lam \times (-N,N)$, and hence
\begin{eqnarray*}
& & \hs{-6mm} \sum_{i=0}^M \mE \int_0^t \int_{\pa D^{\lam_i}} \int_{-N}^N \lam_i (x) (-a (\xi) \cdot {\bf n } (x)) \bar{f}_{1,+}^{(\lam_i)} (s,x,\xi) \bar{f}_{2,-}^{(\lam_i)} (s,x,\xi) d\xi d \si (x) ds \\
& & = \mE \int_0^t \int_{\pa D^{\lam_i}} \int_{-N}^N (-a (\xi) \cdot {\bf n} (x)) \bar{f}_{2,-} \sum_{i=1}^M \lam_i \bar{f}_{1,+}^{(\lam_i)} d\xi d \si (x) ds \\
& & = \mE \int_0^t \int_{\pa D^{\lam_i}} \int_{-N}^N (-a (\xi) \cdot {\bf n} (x)) \bar{f}_{1,+} \bar{f}_{2,-} d\xi d \si (x) ds .
\end{eqnarray*}Here recall that $\bar{f}_{1,+} = \sum_{i=0}^M \lam_i \bar{f}_{1,+}^{(\lam_i)}$. Thus, summing over $i = 0, \ldots , M$ yields
\begin{eqnarray}
& & \hs{-6mm} -\mE \int_{D} \int_{-N}^N f_{1,+}^+ (t , x,\xi ) f_{2,-}^+ (t , x, \xi) d\xi dx \nonumber \\
& & \hs{-6mm} \leq  -\mE \int_{D} \int_{-N}^N f_{1,+}^0 ( x,\xi ) f_{2,-}^0 ( x, \xi) d\xi dx \nonumber \\
& & -\mE \int_0^t \int_{\pa D} \int_{-N}^N (-a (\xi) \cdot {\bf n} (x)) \bar{f}_{1,+} (s, x,\xi ) \bar{f}_{2,-} (s, x, \xi) d\xi d\si (x) ds \nonumber \\
& & + \sum_{i=0}^M ( I_1 + I_2 + I_3 + I_N + \eta_N^t (\ep , \de) + \eta_N^0 (\ep , \de) + r_N (\ep , \de)). \label{summing}
\end{eqnarray}Now note that 
\begin{gather}
 \lim_{\ep , \de \to 0} \sum_{i=0}^M I_2 = -\mE \int_0^t \int_D \int_{-N}^N \nabla (\sum_{i=0}^M \lam_i) a (\xi ) f_{1,+} f_{2,-} d\xi dx ds  =0. \label{fffvvv}
\end{gather}
In a similar way as in the proof of \cite[Theorem 11]{a} we obtain
\begin{gather}
 \lab I_1 \rab \leq C \de \ep^{-1}, \hs{3mm} \lab I_2 \rab \leq C \la \ep^2 \de^{-1} + r (\de) \ra. \label{maenoyatu}
\end{gather}Finally, we compute the boundary term on the right hand side of (\ref{summing}) as follows:
\begin{eqnarray}
& & \hs{-6mm} -\int_{-N}^N (-a \cdot {\bf n}) \bar{f}_{1,+} \bar{f}_{2,-} d\xi \nonumber \\
& & \hs{-6mm} = -\int_{-N}^{u_{2,b}} (-a \cdot {\bf n}) \bar{f}_{1,+} \bar{f}_{2,-} d\xi -\int_{u_{2,b}}^{u_{1,b} \vee u_{2,b}} (-a \cdot {\bf n}) \bar{f}_{1,+} \bar{f}_{2,-} d\xi \nonumber \\
& & -\int_{u_{1,b} \vee u_{2,b}}^N (-a \cdot {\bf n}) \bar{f}_{1,+} \bar{f}_{2,-} d\xi \nonumber \\
& & \hs{-6mm} \leq  -\int_{-N}^{u_{2,b}} \bar{f}_{1,+} \pa_\xi \bar{m}_N^{2,-} d\xi + M_b \int_{u_{2,b}}^{u_{1,b} \vee u_{2,b}} d\xi  -\int_{u_{1,b} \vee u_{2,b}}^N \pa_\xi \bar{m}_N^{1,+} \bar{f}_{2,-} d\xi \nonumber \\
& & \hs{-6mm} \leq - M_b \int_{\mR} f_{1,+}^b f_{2,-}^b d\xi. \label{sugoiyatu}
\end{eqnarray}Now we take $\de = \ep_n^{4/3}$. Letting $\ep_n \to 0$ and then letting $N \to \infty$, we immediately deduce (\ref{contraction2}) from (\ref{summing}), (\ref{fffvvv}), (\ref{maenoyatu}) and (\ref{sugoiyatu}).

\noindent {\it Proof of Corollary \ref{uniqueness}} \ Let $f_+$ be a generalized solution to (\ref{SSCL1})-(\ref{SSCL3}) with the initial datum ${\bf 1}_{u_0 > \xi}$ and the boundary datum ${\bf 1}_{u_b > \xi}$. It follows from Theorem \ref{contraction} and Lemma \ref{abcdefg} that for $t \in [0,T)$,
\begin{gather*}
 \mE \int_D \int_\mR f_+^\pm (t , x, \xi) (1-f_+^\pm (t , x , \xi)) d\xi dx = 0.
\end{gather*}By Fubini's theorem, for $t \in [0,T)$ there is a set $E_t$ of full measure in $\Omega \times D$ such that, for $(\omega , x) \in E_t$, $f_+^\pm (\omega , t, x,\xi ) \in \{ 0,1 \}$ for a.e. $\xi \in \mR$. Since $f_+^\pm (t , x, \xi ) = \nu^\pm_{t,x} (\xi , \infty)$ with a Young measure $\nu^\pm$ on $\Omega \times Q$, there exists $u^\pm (\omega , t,x) \in \mR$ such that $f_+^\pm (\omega , t , x, \xi) = {\bf 1}_{u^\pm (\omega , t,x ) > \xi}$ for a.e. $(\omega , x , \xi )$. This gives that $u^\pm (\omega , t,x)= \int_\mR (f_+^\pm (\omega , t , x,\xi ) - {\bf 1}_{\xi < 0}) d\xi$ and hence $u^\pm$ are predictable. Moreover, (\ref{esti23}) is a direct consequence of (\ref{esti26}). Consequently, we see that $u^+$ (which equals $u^- $ for a.e. $t \in [0,T)$ from Lemma \ref{abcdefg}) is a kinetic solution to (\ref{SSCL1})-(\ref{SSCL3}). \\

\section{Existence}

\indent We state the conditions under which one considers the existence of kinetic solutions.

\begin{description}
\item[${\rm (H_1')}$] The flux function $A$ is of class $C^2$ and the second derivative $A''$ is bounded on $\mR$.
\item[${\rm (H_3')}$] $u_0 \in C^2 (\overline{D})$ and $u_b \in L^\infty (\Si)$ are deterministic. Moreover, $u_b$ is the trace on $\Si$ of a function $U \in C([0,T] \times \overline{D})$ such that $\pa_t U \in C^{\alpha,0} ([0,T] \times D)$, $\Delta U \in C^{\alpha,0} ([0,T] \times \overline{D})$, $U(t, \cdot ) \in W^{2,p} (D)$ for some $\alpha \in (0,1)$ and for any $p>1$. 
\end{description}

It is shown in \cite[Remark 5.1.14]{r} that initial boundary value problem
\begin{eqnarray}
\left\{ \begin{array}{ll} \label{eac}
\pa_t \ti{u}^\ep = \ep \Delta \ti{u}^\ep &  \hs{5mm} {\rm in} \ Q \\
\ti{u}^\ep (0) = 0 & \hs{5mm} {\rm on} \ D \\
\ti{u}^\ep = u_b & \hs{5mm} {\rm on} \ \Sigma \\
\end{array} \right.
\end{eqnarray}admits a unique solution $\ti{u}^\ep \in C ([0,T] \times \ov{D})$. Moreover, all of $\ti{u}^\ep$, $\nabla \ti{u}^\ep$, $\pa_t \ti{u}^\ep$, $\ep \Delta \ti{u}^\ep$ exist and are bounded in $L^\infty (Q)$ uniformly in $\ep \in (0,1]$. \\
\indent The purpose of this section is to prove the following:

\begin{theorem} \label{existence1} {\it Assume that $D$ is a bounded convex domain with $C^2$ boundary. Let the assumptions (${\rm H_1'}$), (${\rm H_2}$) and (${\rm H_3'}$) hold true. Then there exists a unique kinetic solution to} (\ref{SSCL1})-(\ref{SSCL3}).
\end{theorem}

\begin{theorem}[The case of homogeneous Dirichlet boundary condition] \label{existence2} {\it Assume that $D$ is a bounded convex domain with $C^2$ boundary. Assume that $u_0 \in C^2 (\ov{D})$ and $u_b \equiv 0$. If the assumptions (${\rm H_1}$) and (${\rm H_2}$) hold true, then the problem}(\ref{SSCL1})-(\ref{SSCL3}) {\it has a unique kinetic solution.}
\end{theorem}

To prove the theorems we consider the following homogeneous Dirichlet boundary problem
\begin{eqnarray}
\left\{ \begin{array}{ll} \label{approximation}
dv^{\ep} + {\rm div} \hs{0.5mm} A(v^{\ep} + \ti{u}^\ep) dt = \ep \Delta v^{\ep} dt + \Phi^\ep (v^{\ep} + \ti{u}^\ep) dW(t) &  \hs{5mm} {\rm in} \ \Omega \times Q \\
v^{\ep} (0) = u_0 & \hs{5mm} {\rm on} \ \Omega \times  D \\
v^{\ep} = 0 & \hs{5mm} {\rm on} \ \Omega \times  \Sigma \\
\end{array} \right.
\end{eqnarray}where $\Phi^\ep$ is a suitable Lipschitz approximation of $\Phi$ satisfying (\ref{H21}), (\ref{H22}) uniformly in $\ep$. The functions $g_k^\ep$ and $G^{\ep,2}$ will be defined as in the case $\ep = 0 $. By the mean value theorem
\begin{gather*}
 A (r + \ti{u}^\ep) = A (r) + a (r + \theta(t,x) \ti{u}^\ep (t,x) ) \ti{u}^\ep (t,x)
\end{gather*}with $\theta (t,x) \in (0,1) $. Set $g_1 (t,x,r) = a (r + \theta (t,x) \ti{u}^\ep (t,x)) \ti{u}^\ep (t,x) $ and $g_2 (t,r) = A (r)$. By our assumption on the flux $A$ we have that $\lab g_1 (t,x,r) \rab \leq C (1 + \lab r \rab)$, $r \in \mR$, with some constant $C$ independent of $\ep$ and that $g_2$ has at most polynomial growth. Thus, thanks to \cite{h} equation (\ref{approximation}) admits a unique $L^p (D)$ valued continuous solution provided $p$ is large enough and $u_0$ is an $\mF_0$-measurable $L^p (D)$ valued random element. \\
\indent To obtain some energy estimates on equation (\ref{approximation}) we truncate $A$ and $\Phi^\ep$ as follows: Let $A_n (r)$ and $\Phi^\ep_n (r)$ be continuous functions for every integer $n$, such that they are globally Lipschitz, $A_n = A$, $\Phi_n^\ep = \Phi^\ep$ for $\lab r \rab \leq n$, and $A_n = \Phi_n^\ep = 0$ for $\lab r \rab \geq n+1$. Moreover, $A_n$ and $\Phi^\ep_n$ satisfy the same Lipschitz constants and the (same) polynomial growth as $A$ and $\Phi^\ep$, respectively. The functions $g_{n,k}^\ep$ and $G_n^{\ep ,2}$ will be defined as in the case $\ep = 0$. We also chose a sequence of $L^\infty (D) \cap C^\infty (D) $-valued random variables $u_{0n}$ converging to $u_0$ in $L^p$ almost surely. In the same way as in the case of (\ref{approximation}), we have the existence of a unique solution of equation
\begin{gather}
dv_n^{\ep} + {\rm div} \hs{0.5mm} A_n(v_n^{\ep} + \ti{u}^\ep) dt = \ep \Delta v_n^{\ep} dt + \Phi_n^\ep (v_n^{\ep} + \ti{u}^\ep) dW(t) \label{approximation2}
\end{gather}with the Dirichlet boundary condition $v^\ep_n = 0$ on $\Si$ and the initial condition $v^\ep_n (0) = u_{0n}$ on $D$. By virtue of \cite[Lemma 4.3]{h} we have for $p \geq 2$, 
\begin{eqnarray}
& & \hs{-6mm} \lno v_n^\ep (t) \rno_{L^p(D)}^p + \ep p (p - 1) \int^t_0 \int_D \lab v_n^\ep \rab^{p-2} \lab \nabla v_n^\ep \rab^2 dx ds \nonumber \\
& & \hs{-6mm} \leq \lno u_{0n} \rno_{L^p (D)}^p + p \int_0^t \int_D \lab v_n^\ep \rab^{p-2} v_n^\ep A'_n (v_n^\ep + \ti{u}^\ep) \cdot \nabla (v_n^\ep + \ti{u}^\ep) dx ds \nonumber \\
& & + p \sum_{k=1}^\infty \int_0^t \int_D \lab v_n^\ep \rab^{p-2} v_n^\ep g_{n,k}^\ep (x, v_n^\ep + \ti{u}^\ep) dx d\beta_k(s) \nonumber \\
& & + \frac{1}{2} p (p-1) \int_0^t  \int_D \lab v_n^\ep \rab^{p-2}  G_n^{\ep ,2} (x, v_n^\ep + \ti{u}^\ep)  dx ds \hs{10mm} {\rm a.s.}  \label{energy}
\end{eqnarray}Let us consider the second term on the right hand side of (\ref{energy}). By the assumption that $A'' \in L^\infty (\mR) $ and by the Dirichlet boundary condition, one has
\begin{eqnarray}
& & \hs{-6mm} \lab \int_0^t \int_D \lab v_n^\ep \rab^{p-2} v_n^\ep A'_n (v_n^\ep + \ti{u}^\ep) \cdot \nabla (v_n^\ep + \ti{u}^\ep) dx ds \rab \nonumber \\
& & \hs{-6mm} \leq \lab \int_0^t \int_D {\rm div}  \la \int_0^{v_n^\ep} \lab \xi \rab^{p-2} \xi A'_n (\xi + \ti{u}^\ep) d\xi \ra dx ds \rab \nonumber \\
& & + \lab \int_0^t \int_D \int_0^{v_n^\ep} \lab \xi \rab^{p-2} \xi A''_n ( \xi + \ti{u}^\ep ) \cdot \nabla \ti{u}^\ep d\xi dx ds \rab \nonumber \\
& & + \lab \int_0^t \int_D \lab v_n^\ep \rab^{p-2} v_n^\ep A'_n (v_n^\ep + \ti{u}^\ep ) \cdot \nabla \ti{u}^\ep  dx ds \rab \nonumber \\
& & \hs{-6mm} \leq \frac{1}{p} \lno A''_n \rno_{L^\infty} \lno \nabla \ti{u}^\ep \rno_{L^\infty} \int_0^t \int_D \lab v_n^\ep \rab^p dx ds \nonumber \\
& & + C \lno \nabla \ti{u}^\ep \rno_{L^\infty} \int_0^t \int_D \lab v_n^\ep \rab^{p-1} (1 + \lab v_n^\ep + \ti{u}^\ep \rab) dx ds \nonumber \\
& & \hs{-6mm} \leq C \la 1 + \int_0^t \lno v_n^\ep (s) \rno_{L^p (D)}^p ds  \ra ,  \label{energy1}
\end{eqnarray}where and in what follows $C$ denotes various constants which may depend on $p$, $u_0$, $u_b$ and $T$, but not on $\ep$ as well as $n$. By (\ref{H21}) the fourth term is easily estimated:
\begin{eqnarray}
& & \hs{-6mm} \int_0^t \int_D \lab v_n^\ep \rab^{p-2}  G_n^{\ep , 2} (x, v_n^\ep + \ti{u}^\ep ) dx ds \nonumber \\
& & \leq C \la 1+ \int_0^t \lno v_n^\ep (s) \rno_{L^p(D)}^p ds \ra. \label{energy2}
\end{eqnarray}Thus, expectation and application of the Gronwall lemma yield
\begin{gather*}
 \mE \lno v_n^\ep (t) \rno_{L^p (D)}^p \leq C \la 1 + \mE \lno u_{0n} \rno_{L^p (D)}^p \ra .
\end{gather*}Furthermore, by using the Burkholder-Davis-Gundy inequality, we have (see \cite{a} and \cite{m})
\begin{gather}
 \mE \sup_{0 \leq t \leq T} \lno v_n^\ep (t) \rno_{L^p (D)}^p + \ep \mE \int_0^T \int_D \lab v_n^\ep \rab^{p-2} \lab \nabla v_n^\ep \rab^2 dx dt \leq C \label{energy3}
\end{gather}for every $p \geq 2$. Accordingly, by the same argument as in \cite{h}, using the Gy${\rm \ddot{o}}$ngy-Krylov characterization of convergence in probability (see \cite[Lemma 4.1]{ac}), we have that $v_n^\ep$ converges in $C([0,T] ; L^p (D) )$, in probability, to $v^\ep$ as $n \to \infty$. This convergence, together with (\ref{energy3}), deduces that, up to subsequence, $\lab v_n^\ep \rab^{\frac{p-2}{2}} \nabla v_n^\ep$ converges to $\lab v^\ep \rab^{\frac{p-2}{2}} \nabla v^\ep$, as $n \to \infty$, weakly in $L^2 (\Omega \times Q)$. Consequently, passing $n$ to infinity in (\ref{energy1}) yields: For every $p \geq 2$, 
\begin{gather}
\mE \sup_{0 \leq t \leq T} \lno v^\ep \rno_{L^p (D)}^p + \ep \mE \int_0^T \int_D \lab v^\ep \rab^{p-2} \lab \nabla v^\ep \rab^2 dx dt \leq C. \label{energy4}
\end{gather}Next (\ref{energy}), together with (\ref{energy1}) and (\ref{energy2}), gives
\begin{eqnarray*}
& & \hs{-6mm} \ep p (p-1) \int_0^T \int_D \lab v^\ep  \rab^{p-2} \lab \nabla v^\ep \rab^2 dx ds \\
& & \hs{-6mm} \leq \lno u_0 \rno_{L^p (D)}^p + C \la \int_0^T  \int_D \lno v^\ep (s) \rno_{L^p (D)}^p  ds +1  \ra \\
& & +  p \int_0^T \int_D \lab v^\ep \rab^{p-2} v^\ep \Phi^\ep (v^\ep + \ti{u}^\ep) dx dW(s).
\end{eqnarray*}Taking the square, then expectation, we deduce by the It${\rm \hat{o}}$ isometry
\begin{eqnarray*}
& & \hs{-6mm} \mE \lab \int_0^T \int_D \ep \lab v^\ep \rab^{p-2} \lab \nabla v^\ep \rab^2 dx ds  \rab \\
& & \hs{-6mm} \leq C \lno u_0 \rno_{L^p (D)}^{2p} + C \mE \la \int_0^T \lno v^\ep (s) \rno_{L^p (D)}^p ds + 1 \ra^2 \\
& & + \mE \int_0^T \sum_{k=1}^\infty \lab \int_D \lab v^\ep \rab^{p-2} v^\ep g_k^\ep (x , v^\ep + \ti{u}^\ep) dx \rab^2 dt 
\end{eqnarray*}By (\ref{H21}), (\ref{energy4}) and the Cauchy-Schwartz inequality,
\begin{gather}
\mE \lab \int_0^T \int_D \ep \lab v^\ep \rab^{p-2} \lab  \nabla v^\ep \rab^2 dx ds  \rab^2 \leq C. \label{bounded}
\end{gather}Define an $L^p (D)$-valued process $u^\ep (t) $ by $u^\ep (t) = v^\ep (t) + \ti{u}^\ep (t)$. It satisfies: a.s., for all $t \in [0,T]$ and $\phi \in C_c^\infty (D)$,
\begin{eqnarray*}
& &  \int_D u^\ep \phi dx = \int_D u_0 \phi dx + \int_0^t \int_D A (u^\ep ) \nabla \phi dx ds \\
& & \hs{10mm} + \ep \int_0^t \int_D u^\ep \Delta \phi dx ds + \sum_{k=1}^\infty \int_0^t \int_D g_k^\ep (x, u^\ep) \phi dx d\beta_k (s),
\end{eqnarray*}in short,
\begin{gather*}
d u^\ep + {\rm div} A (u^\ep) dt - \ep \Delta u^\ep dt = \Phi^\ep (u^\ep ) dW (t) .
\end{gather*}By the smooth approximations as in \cite[Proposition 18]{a} and \cite{aa} (also see \cite{ad} and \cite{ab}) or by the generalized It\^o formula in \cite[Proposition A.1]{t} we obtain that $f^\ep = {\bf 1}_{u^\ep > \xi} $ satisfies the kinetic formulation, more precisely, a.s., for all $\ph \in C_c^\infty ([0,T) \times D \times \mR)$
\begin{eqnarray}
& & -\int_0^T \int_D \int_\mR f^\ep \pa_t \ph d\xi dx dt = \int_D \int_\mR f_0 \ph (0) d\xi dx \nonumber \\
& & + \int_0^T \int_D \int_\mR a (\xi ) f^\ep \nabla \ph d\xi dx ds + \ep \int_0^T \int_D \int_\mR f^\ep \Delta \ph d\xi dx dt \nonumber \\
& & + \sum_{k=1}^\infty \int_0^T \int_D \int_\mR g_k^\ep \ph d \nu_{t,x}^\ep (\xi) dx d\beta_k (t) \nonumber \\
& & + \frac{1}{2} \int_0^T \int_D \int_\mR \pa_\xi \ph G_\ep^2 d \nu_{t,x}^\ep (\xi ) dx dt - m^\ep (\pa_\xi \ph), \label{formulation}
\end{eqnarray}where $f_0 (\xi) = {\bf 1}_{u_0 > \xi}$, $\nu_{t,x}^\ep =- \de_{u^\ep (t,x)} $ and
\begin{gather*}
 m^\ep (\phi ) = \int_0^T \int_D \phi (t,x,u^\ep (t,x)) \ep \lab \nabla u^\ep \rab^2 dx dt 
\end{gather*}for $\phi \in C_b ([0,T] \times D \times \mR)$. It follows from (\ref{energy4}) and the uniform boundedness of $\ti{u}^\ep$ in $\ep$ that for $t \in [0,T]$,
\begin{gather}
\mE \int_D \int_\mR \lab \xi \rab^p d \nu_{t,x}^\ep (\xi) dx \leq C. \label{timeuni}
\end{gather}We need the following compactness result. (For the proof see \cite{a}.)

\begin{theorem} \label{T43} {\it Let $\{ \ep_n \} \downarrow 0$. Suppose}
\begin{gather}
\lim_{R \to \infty} \sup_n \mE \int_0^T \int_D \int_{B_R^c} d\nu_{t,x}^{\ep_n} (\xi) dx dt = 0. \label{kyokugenn} 
\end{gather}
\begin{description}
\item[${\rm (a)}$] \hs{2mm} {\it There exist a Young measure $\nu$ on $\Omega \times Q$ and a subsequence still denoted $ \{ \nu^{\ep_n} \}$ such that for all $h \in L^1 (\Omega \times Q)$, for all $\phi \in C_b (\mR)$ }
 \begin{eqnarray}
  & & \lim_{n \to \infty} \mE \int_0^T \int_D h (t,x) \int_\mR \phi (\xi) d \nu_{t,x}^{\ep_n} (\xi) dx dt \nonumber \\
  & & \hs{4mm} = \mE \int_0^T \int_D h(t,x) \int_\mR \phi (\xi) d\nu_{t,x} (\xi) dx dt.  \label{shuusoku}
 \end{eqnarray}
\item[${\rm (b)}$] \hs{2mm} {\it There exist a kinetic function $f$ on $\Omega \times Q \times \mR$ and a subsequence still denoted $\{ f^{\ep_n} \}$ such that $f^{\ep_n} \rightharpoonup f $ in $L^\infty (\Omega \times Q \times \mR)$-weak*.}
\end{description}
\end{theorem}Here we take notice that as stated in \cite{aa} we may assume that the $\sigma$-algebra $\mathscr{F}$ is countably generated and hence $L^1 (\Omega \times Q)$ is separable. \\

\noindent {\it Proof of Theorem \ref{existence1}} \ For $\de > 0$ sufficiently small we define 
\begin{gather*}
s(x) = 
\left\{ \begin{array}{ll}
\min \{ {\rm dist} (x , \pa D) , \de \} &  \hs{1mm} {\rm for} \ x \in D, \\
- \min \{ {\rm dist} (x,\pa D) , \de  \} & \hs{1mm} {\rm for} \ x \in \mR^d \setminus D . \\
\end{array} \right.
\end{gather*}The function $s$ is Lipschitz continuous in $\mR^d$ and smooth on the closure of $\{ x \in \mR^d ; \lab s(x) \rab < \de \}$. For $\de > 0$ define the function $\Theta_\ep$ by 
\begin{gather*}
\Theta_\ep (x) = 1 - {\rm exp} \la - \frac{M+\ep L}{\ep} s(x) \ra,
\end{gather*}where $M > 0$ and $L = \sup_{0 < s(x) < \de} \lab \Delta s(x) \rab$. This function satisfies the weak differential inequality
\begin{gather}
 M \int_D \lab \nabla \Theta_\ep \rab \phi dx \leq \ep \int_D \nabla \Theta_\ep \cdot \nabla \phi dx + (M + \ep L) \int_{\pa D} \phi d\si \label{kkklll}
\end{gather}for any $\phi \in C_c^\infty (\mR^d) $. (See \cite[p. 129]{g}.) Let $N > 0$ and let us fix any non-negative $\ph \in C_c^\infty ([0,T) \times \mR^d \times (-N,N))$. We regularize (\ref{formulation}) by convolution. Namely, we apply (\ref{formulation}) to the test function $(\Theta^+_\ep \ph^\lam )* \rho^\lam_\eta$, where $\Theta^+_\ep (x) = \max \{ \Theta_\ep (x), 0 \}$, $\ph^\lam = \ph \lam$ and $(\rho^\lam_\eta)$ denotes the right-decentered approximations to the identity on $\mR^d$, to obtain
\begin{eqnarray}
& & \hs{-6mm} -\int_0^T \int_D \int_\mR f^{\ep ,\eta} \Theta^+_\ep \pa_t \ph^\lam d\xi dx dt \nonumber \\
& & \hs{-6mm} = \int_D \int_\mR f_0^\eta \Theta^+_\ep \ph^\lam (0) d\xi dx + \int_0^T \int_D \int_\mR a (\xi ) f^{\ep ,\eta} \nabla ( \Theta^+_\ep \ph^\lam ) d\xi dx dt \nonumber \\ 
& & + \ep \int_0^T \int_D \int_\mR f^{\ep ,\eta} \Delta (\Theta^+_\ep \ph^\lam) d\xi dx dt \nonumber \\
& & + \sum_{k=1}^\infty \int_0^T \int_D \int_\mR g_k^\ep (\Theta^+_\ep \ph^\lam )* \rho^\lam_\eta  d \nu_{t,x}^\ep (\xi) dx d\beta_k (t) \nonumber \\
& & + \frac{1}{2} \int_0^T \int_D \int_\mR \pa_\xi (\Theta^+_\ep \ph^\lam) * \rho^\lam_\eta G_\ep^2 d \nu_{t,x}^\ep (\xi ) dx dt  - m^\ep (\pa_\xi (\Theta^+_\ep \ph^\lam)* \rho^\lam_\eta), \nonumber
\end{eqnarray}where $f^{\ep , \eta} = f^\ep * \check{\rho}^\lam_\eta$, $f_0^\eta = f_0 * \check{\rho}^\lam_\eta $ and $\check{\rho}^\lam_\eta (x) = \rho^\lam_\eta (-x)$. Setting 
\begin{gather*}
M_N = \sup_{\xi \in [-N,N]} \lab a(\xi) \rab,
\end{gather*}we calculate
\begin{eqnarray*}
& & \int_0^T \int_D \int_\mR a (\xi) f^{\ep , \eta} \nabla (\Theta^+_\ep \ph^\lam) d \xi dx dt + \ep \int_0^T \int_D \int_\mR f^{\ep , \eta} \Delta (\Theta^+_\ep \ph^\lam) d\xi dx dt \\
& & \leq M_N \int_0^T \int_D \int_\mR f^{\ep , \eta} \lab \nabla \Theta^+_\ep \rab \ph^\lam d \xi dx dt + \int_0^T \int_D \int_\mR a (\xi) f^{\ep , \eta} \Theta^+_\ep \nabla \ph^\lam d \xi dx dt \\
& & + \ep \int_0^T \int_D \int_\mR \left( - \nabla ( f^{\ep , \eta} \ph^\lam ) \nabla \Theta^+_\ep + 2 f^{\ep, \eta} \nabla \Theta^+_\ep \nabla \ph^\lam + \Theta_\ep^+ \Delta \ph^\lam \right) d\xi dx dt \\
& & \leq (M_N + \ep L) \int_{\pa D} f_b^\eta \ph^\lam d\si +  \int_0^T \int_D \int_\mR a (\xi) f^{\ep , \eta} \Theta^+_\ep \nabla \ph^\lam d \xi dx dt \\
& & + \ep \int_0^T \int_D \int_\mR \left( 2 f^{\ep, \eta} \nabla \Theta^+_\ep \nabla \ph^\lam + \Theta_\ep^+ \Delta \ph^\lam \right) d\xi dx dt.
\end{eqnarray*}Here we used (\ref{kkklll}) with $M$ and $\ph$ replaced by $M_N$ and $f^{\ep , \eta} \ph^\lam$, respectively. Letting $\eta \downarrow 0$ and summing over $i$, we obtain 
\begin{eqnarray}
& & \hs{-6mm} -\int_0^T \int_D \int_{-N}^N \Theta_\ep f^{\ep} ( \pa_t + a(\xi) \cdot \nabla) \ph d\xi dx dt \nonumber \\
& & \hs{3mm} - \int_D \int_{-N}^N \Theta_\ep f_0 \ph (0) d\xi dx - (M_N + \ep L) \int_0^T \int_{\pa D} \int_{-N}^N f^b \ph d\xi d\si dt \nonumber \\
& & \hs{-6mm} \leq \ep \int_0^T \int_D \int_{-N}^N \Theta_\ep f^{\ep} \Delta \ph d\xi dx dt + 2\ep \int_0^T \int_D \int_{-N}^N f^\ep \nabla \ph \cdot \nabla \Theta_\ep d\xi dx dt \nonumber \\
& & \hs{3mm} + \sum_{k=1}^\infty \int_0^T \int_D \int_{-N}^N \Theta_\ep g_k^\ep \ph d \nu_{t,x}^\ep (\xi) dx d\beta_k (t) \nonumber \\
& & \hs{3mm} + \frac{1}{2} \int_0^T \int_D \int_{-N}^N \Theta_\ep \pa_\xi \ph G_\ep^2 d \nu_{t,x}^\ep (\xi ) dx dt - m^\ep (\Theta_\ep \pa_\xi \ph ). \label{kinetic futousiki}
\end{eqnarray}Note here that $\sup_{\ep \in (0,1]} \int_D \lab \nabla \Theta_\ep \rab dx \leq C < \infty $ by (\ref{kkklll}). We will pass $\ep $ to $0$ in (\ref{formulation}) and (\ref{kinetic futousiki}) through a subsequence $\{ \ep_n \}$. By virtue of (\ref{bounded}) with $p=2$ we have, up to subsequence, $m^{\ep_n}$ converges to a kinetic measure $m$ in $L^2_w (\Omega ; {\mathcal M}_b)$-weak*, where $L^2_w (\Omega ; \mathcal{M}_b)$ is the space of all weak*-measurable mappings $m: \Omega \to \mathcal{M}_b$ with $\mathbb{E} \| m \|_{\mathcal{M}_b}^2 < \infty$. It is shown in \cite{a} that $m$ satisfies (\ref{vanish}) and the process $t \mapsto \int_{(0,t) \times D \times \mR} \ph (x,\xi) dm (s,x,\xi) $ is predictable for any $\phi \in C_b (D \times \mR)$. \\
\indent Since (\ref{timeuni}) immediately implies (\ref{kyokugenn}), by Theorem \ref{T43} there exist a kinetic function $f_+$, a Young measure $\nu_{t,x}$ and a kinetic measure $m$ which satisfy: For any $\phi \in C_c^\infty ([0,T) \times D \times (-N,N))$, 
\begin{eqnarray}
    & & \hs{-8mm} -\int_0^T \int_D \int_{-N}^N f_+ (\pa_t + a \cdot \nabla) \phi \hs{0.5mm} d\xi dx dt - \int_{D}  \int_{-N}^N f^0_+ \phi(0) \hs{0.5mm} d\xi dx \nonumber \\
    & & \hs{-8mm} =  \sum_{k = 1}^\infty \int_0^T \int_D \int_{-N}^N g_k \hh \phi \hh d \nu_{t,x} (\xi) dx d\beta_k(t) \nonumber \\
    & & \hs{-8mm} +\frac{1}{2} \int_0^T \int_D \int_{-N}^N \pa_{\xi} \phi \hh G^2 \hh d \nu_{t,x} (\xi) dx dt \nonumber \\
    & &  \hs{-8mm} -\int_{(0,T) \times D \times (-N,N)} \pa_{\xi} \phi \hs{0.5mm} dm(t,x,\xi),  \hs{2.5mm} {\rm a.s.}, \label{GKIFO4}
\end{eqnarray}and for any $\ph \in C_c^\infty ([0,T) \times \mR^d \times (-N,N))$ with $\ph \geq 0$, 
\begin{eqnarray}
    & & -\int_0^T \int_D \int_{-N}^N f_+ (\pa_t + a \cdot \nabla) \ph \hs{0.5mm} d\xi dx dt \nonumber \\
    & &  - \int_{D}  \int_{-N}^N f^0_+ \ph(0) \hs{0.5mm} d\xi dx - M_N \int_0^T \int_{\pa D} \int_{-N}^N f^b \ph d\xi d\si dt \nonumber \\
    & & \leq  \sum_{k = 1}^\infty \int_0^T \int_D \int_{-N}^N g_k \hh \ph \hh d \nu_{t,x} (\xi) dx d\beta_k(t) \nonumber \\
    & &  +\frac{1}{2} \int_0^T \int_D \int_{-N}^N \pa_{\xi} \ph \hh G^2 \hh d \nu_{t,x} (\xi) dx dt \nonumber \\
    & &  -\int_{(0,T) \times D \times (-N,N)} \pa_{\xi} \ph \hs{0.5mm} dm(t,x,\xi),  \hs{2.5mm} {\rm a.s.}, \label{GKIFO5}
\end{eqnarray} \indent Now, take any $\ph \in C_c^\infty ([0,T) \times \mR^d \times (-N,N))$. Let $\lam_i$ be an element of the partition of unity on $\ov{D}$. For $i \geq 1 $ and $\eta >0$, define $\bar{\Theta}_\eta (x) = \int_0^{x_d - h_\lam (\bar{x})} \psi_\eta (r - \eta (L_\lam + 1)) dr $, where we have again dropped the index $i$ of $\lam_i$. We apply (\ref{GKIFO4}) to the test function $\phi = \bar{\Theta}_\eta \ph^\lam$ and let $\eta \downarrow 0$ in the resultant equality. Then
\begin{eqnarray}
    & & \hs{-11mm} -\int_0^T \int_{D^\lam} \int_{-N}^N f_+ (\pa_t + a \cdot \nabla) \ph^\lam \hs{0.5mm} d\xi dx dt - \int_{D^\lam}  \int_{-N}^N f^0_+ \ph^\lam(0) \hs{0.5mm} d\xi dx \nonumber \\
    & &  - \int_0^T \int_{\Pi^\lam} \int_{-N}^N (-a (\xi) \cdot {\bf n} (\bar{x}) ) \bar{f}^{(\lam)}_+ \ph^\lam d\xi d \bar{\si} (\bar{x}) dt  \nonumber \\
    & & \hs{-11mm} =  \sum_{k = 1}^\infty \int_0^T \int_{D^\lam} \int_{-N}^N g_k \hh \ph^\lam \hh d \nu_{t,x} (\xi) dx d\beta_k(t) \nonumber \\
    & &  +\frac{1}{2} \int_0^T \int_{D^\lam} \int_{-N}^N \pa_{\xi} \ph^\lam \hh G^2 \hh d \nu_{t,x} (\xi) dx dt \nonumber \\
    & &  -\int_{(0,T) \times D^\lam (-N,N)} \pa_{\xi} \ph^\lam \hs{0.5mm} dm(t,x,\xi),  \hs{5mm} {\rm a.s.}, \label{GKIFO6}
\end{eqnarray}where recall that $\bar{f}^{(\lam)}_+$ denotes any weak* limit of $\int_{D^\lam} f_+ (t,y,\xi ) \rho_\eta (x_d - h_\lam (\bar{x})) dx $ as $\eta \downarrow 0$ in $L^\infty (\Omega \times \Si^\lam \times \mR)$. Since $L^1 (\Omega \times \Si^\lam \times \mR)$ is separable as was mentioned before, the predictability is stable under the weak* topology of $L^\infty (\Omega \times \Si^\lam \times \mR)$. In particular, $f^{(\lam)}_\pm$ is predictable. Combining (\ref{GKIFO5}) with (\ref{GKIFO6}) yields 
\begin{gather*}
\int_0^T \int_{\Pi^\lam} \int_{-N}^N (M_N f^b + (a \cdot {\bf n}) \bar{f}^{(\lam)}_+  ) \bar{\ph} d\xi d \bar{\si} ( \bar{x} ) dt \geq 0
\end{gather*}for all $\bar{\ph} \in C_c^\infty ([0,T ) \times \Pi^\lam \times (-N,N))$ with $\bar{\ph} \geq 0$, and hence $M_N f^b + (a \cdot {\bf n}) \bar{f}^{(\lam)}_+ \geq 0 $ for a.e. $(t,\bar{x}, \xi) \in (0,T) \times \Pi^\lam \times (-N,N) $. For $N > \lno u_b \rno_{L^\infty}$ we set 
\begin{gather}
 \bar{m}_N^{+,\lam} (t,x,\xi ) = M_N (u_b (t,x) - \xi)^+ - \int_\xi^N (-a (\eta)\cdot {\bf n} (x)) \bar{f}^{(\lam)}_+ (t,x,\eta) d\eta
\end{gather}for $ ( t,\bar{x} , \xi ) \in \Si^\lam \times (-N,N) $. Clearly, $\bar{m}^+_N \geq 0$ and $\bar{m}^+_N (N) = 0$. Since $(-a (\xi) \cdot {\bf n}) \bar{f}^{(\lam)}_+ = M_b f_+^b + \pa_\xi \bar{m}_N^+$, it follows from (\ref{GKIFO6}) that
\begin{eqnarray}
    & & \hs{-11mm} -\int_0^T \int_{D^\lam} \int_{-N}^N f_+ (\pa_t + a \cdot \nabla) \ph^\lam \hs{0.5mm} d\xi dx dt \nonumber \\
    & &  - \int_{D^\lam}  \int_{-N}^N f^0_+ \ph^\lam(0) \hs{0.5mm} d\xi dx  - M_N \int_0^T \int_{\Pi^\lam} \int_{-N}^N f_+^b \ph^\lam d\xi d \bar{\si} ( \bar{x} ) dt  \nonumber \\
    & & \hs{-11mm} =  \sum_{k = 1}^\infty \int_0^T \int_{D^\lam} \int_{-N}^N g_k \hh \ph^\lam \hh d \nu_{t,x} (\xi) dx d\beta_k(t) \nonumber \\
    & &  +\frac{1}{2} \int_0^T \int_{D^\lam} \int_{-N}^N \pa_{\xi} \ph^\lam \hh G^2 \hh d \nu_{t,x} (\xi) dx dt \nonumber \\
    & &  -\int_{(0,T) \times D \times (-N,N)} \pa_{\xi} \ph^\lam \hs{0.5mm} dm - \int_0^T \int_{\Pi^\lam} \int_{-N}^N \pa_\xi \ph^\lam d\bar{m}^{+,\lam}_N,  \hs{5mm} {\rm a.s.}. \nonumber
\end{eqnarray}By summing over $i = 0,1, \ldots , M$ we obtain the kinetic formulation (\ref{GKIFO}) for $f_+$ with
\begin{eqnarray*}
 \bar{m}_N^+ (t,x,\xi) = M_N (u_b (t,x) - \xi)^+ - \int_\xi^N (-a(\eta) \cdot {\bf n} (x) ) \bar{f}_+ (t,x,\eta) d\eta.
\end{eqnarray*}
In a similar manner, we can also obtain the kinetic formulation (\ref{GKIFO}) for $f_-$. \\

\noindent {\it Proof of Theorem \ref{existence2}} \ When $u_b = 0$, $\ti{u}^\ep$ becomes $0$ identically. Therefore, there is no need to assume the boundedness of $A''$ in the argument of the previous theorem. \\

\noindent {\bf Acknowledgment} \ The authors wish to thank the referees for their comments helping to improve the manuscript.


\begin{thebibliography}{99}

\bibitem{k} C. Bauzet, G. Vallet, P. Wittbold, The Dirichlet problem for a conservation law with a multiplicative stochastic perturbation, J. Funct. Anal. 266 (2014) 2503-2545.
\bibitem{f} C. Bardos, A.Y. Le Roux, J.-C. N${\rm \acute{e}}$d${\rm \acute{e}}$lec, First order quasilinear equations with boundary condition, Comm. Partial Differential Equations 4 (1979) 1017-1034.
\bibitem{n} G.-Q. Chen, Q. Ding, K.H. Karlsen, On nonlinear stochastic balance laws, Arch. Ration. Mech. Anal. 204 (3) (2012) 707-743.
\bibitem{c} G. Da Prato, J. Zabczyk, Stochastic Equations in Infinite Dimensions, Encyclopedia Math. Appl., vol. 44, Cambridge University Press, Cambridge, 1992. 
\bibitem{aa} A. Debussche, S. De Moor, M. Hofmanov\'a, A regularity result for quasilinear stochastic partial differential equations of parabolic type, to appear in SIAM J. Math. Anal.
\bibitem{t} A. Debussche, M. Hofmanov${\rm \acute{a}}$, J. Vovelle, Degenerate parabolic stochastic partial differential equations: quasilinear case, arXiv: 1309. 5817 [math. A8].
\bibitem{a} A. Debussche, J. Vovelle, Scalar conservation laws with stochastic forcing, J. Funct. Anal. 259 (4) (2010) 1014-1042. 
\bibitem{ad} A. Debussche, J. Vovelle, Scalar conservation laws with stochastic forcing, revised version (2014), http://math.univ-lyon1.fr/\verb|~|vovelle/DebusscheVovelleRevised.pdf.
\bibitem{d} J. Feng, D. Nualart, Stochastic scalar conservation laws, J. Funct. Anal. 255 (2) (2008) 313-373. 
\bibitem{ac} I. Gy${\rm \ddot{o}}$ngy, Existence and uniqueness results for semilinear stochastic partial differential equations, Stochastic Process. Appl. 73 (1998) 271-299.
\bibitem{j} I. Gy${\rm \ddot{o}}$ngy, N. Krylov, Existence of strong solutions for It${\rm \hat{o}}$'s stochastic equations via approximations, Probab. Theory Related Fields 105 (2) (1996) 143-158.
\bibitem{h} I. Gy${\rm \ddot{o}}$ngy, C. Rovira, On $L^p$-solutions of semilinear stochastic partial differential equations, Stochastic Process. Appl. 90 (1) (2000) 83-108.
\bibitem{m} M. Hofmanov${\rm \acute{a}}$, Degenerate parabolic stochastic partial differential equations, Stoch. Pr. Appl. 123 (12) (2013) 4294-4336.
\bibitem{ab} M. Hofmanov\'a, Strong solutions of semilinear stochastic partial differential equations, NoDEA Nonlinear Differential Equations Appl. 20 (3) (2013) 757-778.
\bibitem{b} C. Imbert, J. Vovelle, A kinetic formulation for multidimensional scalar conservation laws with boundary conditions and applications, SIAM J. Math. Anal. 36 (2004) 214-232.
\bibitem{l} J.U. Kim, On a stochastic scalar conservation law, Indiana Univ. Math. J. 52 (1) (2003) 227-256.
\bibitem{p} K. Kobayasi, A kinetic approach to comparison properties for degenerate parabolic-hyperbolic equations with boundary conditions, J. Differential Equations 230 (2006) 682-701.
\bibitem{e} P.L. Lions, B. Perthame, E. Tadmor, A kinetic formulation of multidimensional scalar conservation laws and related equations, J. Amer. Math. Soc. 7 (1) (1994) 169-191. 
\bibitem{r} A. Lunardi, Analytic Semigroups and Optimal Regularity in Parabolic Problems, Birkh${\rm \ddot{a}}$user, Basel (1995).
\bibitem{g} J. M${\rm \acute{a}}$lek, J. Ne${\rm \check{c}}$as, M. Rokyta, M. R\r{u}${\rm \check{z}}$i${\rm \check{c}}$ka, Weak and measure-valued solutions to evolutionary PDEs, Chapman \verb|&| Hall, London, Weinheim, New York, 1996.
\bibitem{o} A. Michel, J. Vovelle, Entropy formulation for parabolic degenerate equations with general Dirichlet boundary conditions and application to the convergence of FV methods, SIAM J. Numer. Anal. 41 (2003) 2262-2293.
\bibitem{i} F. Otto, Initial-boundary value problem for a scalar conservation law, C. R. Acad. Sci. Paris S${\rm \acute{e}}$r. I Math., 322 (1996), 729-734.
\bibitem{q} D. Revuz, M. Yor, Continuous Martingales and Brownian Motion, third ed., Grundlehren Math. Wiss. (Fundamental Principles of Mathematical Sciences), vol. 293, Springer-Verlag, Berlin, 1999.
\bibitem{s} G. Vallet, P. Wittbold, On a stochastic first-order hyperbolic equation in a bounded domain, Infin. Dimens. Anal. Quantum Probab. 12 (4) (2009) 1-39.

\end{thebibliography}
\end{document}